\documentclass[journal]{IEEEtran}
\usepackage{amsmath,amscd}
\usepackage{graphicx}
\usepackage{algorithm}
\usepackage{algorithmic}
\usepackage{amssymb,array}
\usepackage{multirow}
\usepackage{multicol}
\usepackage{amsthm}

\usepackage{enumitem}
\usepackage{hyperref}
\usepackage{bm}
\usepackage{xcolor}
\usepackage[font=small,labelsep=period]{caption}
\usepackage{subcaption}
\usepackage{url}

\def\x{\bm{x}}
\def\y{\bm{y}}
\def\n{\bm{n}}

\def\A{\bm{A}}

\def\R{\bm{R}}

\def\v{\bm{v}}

\def\v{\bm{v}}
\def\W{\bm{W}}
\def\Z{\bm{Z}}
\def\N{\bm{N}}
\def\I{\bm{I}}
\def\V{\bm{V}}

\def\U{\bm{U}}

\def\X{\bm{X}}

\def\S{\bm{S}}
\newcommand{\E}{\mathrm{E}}
\newtheorem{theorem}{Theorem}
\newtheorem{lemma}{Lemma}
\newtheorem{assumption}{Assumption}

\newtheorem{definition}{Definition}
\newtheorem{remark}{Remark}

\makeatletter

\newcolumntype{"}{@{\hskip\tabcolsep\vrule width 1pt\hskip\tabcolsep}}
\makeatother

\newcommand{\mr}{\mathrm}
\newcommand{\BE}{\begin{equation}}
\newcommand{\EE}{\end{equation}}
\newcommand{\BS}{\begin{subequations}}
\newcommand{\ES}{\end{subequations}}

\begin{document}

\title{TARM: A Turbo-type Algorithm for Affine Rank Minimization}

\author{Zhipeng~Xue,~Xiaojun~Yuan,~\IEEEmembership{Senior~Member,~IEEE}, Junjie~Ma,~and~Yi~Ma,~\IEEEmembership{Fellow,~IEEE}}

\maketitle

\begin{abstract}
The affine rank minimization (ARM) problem arises in many real-world applications. The goal is to recover a low-rank matrix from a small amount of noisy affine measurements. The original problem is NP-hard, and so directly solving the problem is computationally prohibitive. Approximate low-complexity solutions for ARM have recently attracted much research interest. In this paper, we design an iterative algorithm for ARM based on message passing principles. The proposed algorithm is termed turbo-type ARM (TARM), as inspired by the recently developed turbo compressed sensing algorithm for sparse signal recovery. We show that, when the linear operator for measurement is right-orthogonally invariant (ROIL), a scalar function called state evolution can be established to accurately predict the behaviour of the TARM algorithm. We also show that TARM converges much faster than the counterpart algorithms for low-rank matrix recovery. We further extend the TARM algorithm for matrix completion, where the measurement operator corresponds to a random selection matrix. We show that, although the state evolution is not accurate for matrix completion, the TARM algorithm with carefully tuned parameters still significantly outperforms its counterparts.
\end{abstract}


\IEEEpeerreviewmaketitle

\section{Introduction}
Low-rank matrices have extensive applications in real-world applications including remote sensing, recommendation systems, global positioning, and system identification. In these applications, a fundamental problem is to recover an unknown matrix from a small number of observations by exploiting its low-rank property \cite{candes2010matrix, davenport2016overview}. In specific, we consider a rank-$r$ matrix $\X^{\ast}\in \mathbb{R}^{n_1\times n_2}$ with the integers $r, n_1,$ and $n_2$ satisfying $r\ll n_1$ and $r\ll n_2$. We aim to recover $\X^{\ast}$ from an affine measurement given by
\begin{align}
		\y=\mathcal{A}(\X^{\ast})\in \mathbb{R}^{m}
\end{align}
where $\mathcal{A}: \mathbb{R}^{n_1\times n_2}\rightarrow \mathbb{R}^m$ is a linear map with $m<n_1 n_2=n$. When $\mathcal{A}$ is a general linear operator such as Gaussian operators and partial orthogonal operators, we refer to the problem as {\it low-rank matrix recovery}; when $\mathcal{A}$ is a selector that outputs a subset of the entries of $\X^{\ast}$, we refer to the problem as {\it matrix completion}. 

The problem can be formally cast as affine rank minimization (ARM):
\begin{align}
	\begin{split}
		\min_{\X}& \text{ rank}(\X)\\
		\text{s.t. }& \y=\mathcal{A}(\X)\label{arm}.
	\end{split}
\end{align}
Problem (\ref{arm}) is NP-hard, and so solving (\ref{arm}) is computationally prohibitive. To reduce complexity, a popular alternative to (\ref{arm}) is the following nuclear norm minimization (NNM) problem:
\begin{align}
	\begin{split}
		\min_{\X}&\  \|\X\|_{\ast}\\
		\text{s.t. }& \y=\mathcal{A}(\X).\label{nuclear_p}
	\end{split}
\end{align}
In \cite{recht2010guaranteed}, Recht et al. proved that when the restricted isometry property (RIP) holds for the linear operator $\mathcal{A}$, the ARM problem in (\ref{arm}) is equivalent to the NNM problem in (\ref{nuclear_p}). The NNM problem can be solved by semidefinite programing  (SDP). Existing convex solvers, such as the interior point method \cite{liu2009interior}, can be employed to find a solution in polynomial time. However, SDP is still computationally involving, especially when applied to large-scale problems with high dimensional data. To address this issue, low-cost iterative methods, such as the singular value thresholding (SVT) method \cite{cai2010singular} and the proximal gradient algorithm \cite{toh2010accelerated}, have been proposed to further reduce complexity at the cost of a certain amount of performance degradation.

In real-world applications, perfect measurements are rare, and noise is naturally introduced in the measurement process. That is, we want to recover $\X^{\ast}$ from a noisy measurement of
\begin{align}
		\y=\mathcal{A}(\X^{\ast})+\n\label{Gaussian_mea}
\end{align}
where $\n \in \mathbb{R}^m$ is a Gaussian noise with zero-mean and covariance $\sigma^2\bm{I}$ and is independent of $\mathcal{A}(\X^{\ast})$. To recover the low-rank matrix $\X^{\ast}$ from (\ref{Gaussian_mea}), we turn to the following stable formulation of the ARM problem:
\begin{align}
	\begin{split}
		\min_{\X}& \|\y-\mathcal{A}(\X)\|_2^2\\
		\text{s.t. }& \text{ rank}(\X)\leq r.
	\end{split}\label{rank_constrain}
\end{align}
The problem in (\ref{rank_constrain}) is still NP-hard and difficult to solve. Several suboptimal algorithms have been proposed to yield approximate solutions to (\ref{rank_constrain}). For example, the author in \cite{jain2013low} proposed an alternation minimization method to factorize rank-$r$ matrix $\X^{\ast}$ as the product of two matrices with dimension $n_1\times r$ and $r\times n_2$ respectively. This method is more efficient in storage than SDP and SVT methods, especially when large-dimension low-rank matrices are involved. A second approach borrows the idea of iterative hard thresholding (IHT) for compressed sensing. For example, the singular value projection (SVP) algorithm \cite{jain2010guaranteed} for stable ARM can be viewed as a counterpart of the IHT algorithm \cite{blumensath2009iterative} for compressed sensing. SVP solves the stable ARM problem by combining the projected gradient method with singular value decomposition (SVD). Improved version of SVP, termed normalized IHT (NIHT) \cite{tanner2013normalized}, adaptively selects the step size of the gradient descent step of SVP, rather than uses a fixed step size. These algorithms involve a projection step which projects a matrix into a low-rank space using truncated SVD. In \cite{wei2016guarantees}, a Riemannian method, termed RGrad, was proposed to extend NIHT by projecting the search direction of gradient descent into a low dimensional space. Compared with the alternation minimization method, these IHT-based algorithms exhibit better convergence performance with lower computational complexity. Furthermore, the convergence of these IHT-based algorithms are guaranteed when a certain restricted isometry property (RIP) holds \cite{jain2010guaranteed, tanner2013normalized,wei2016guarantees}.

In this paper, we aim to design low-complexity iterative algorithms to solve the stable ARM problem based on message-passing principles \cite{ma2014turbo}, a perspective different from the existing approaches mentioned above. In specific, we propose a novel turbo-type algorithm, termed turbo-type affine rank minimization (TARM), for solving the stable ARM problem, as inspired by the turbo compressed-sensing (Turbo-CS) algorithm for sparse signal recovery \cite{ma2014turbo, xue2017denoising}. Interestingly, although TARM is designed based on the idea of message passing, the resultant algorithm bears a similar structure to the gradient-based algorithms such as SVP and NIHT. A key difference of TARM from SVP and NIHT resides in an extra step in TARM for the calculation of the so-called extrinsic messages. With this extra step, TARM is able to find a better descendent direction for each iteration, so as to achieve a much higher convergence rate than SVP and NIHT. For low-rank matrix recovery, we establish a state evolution technique to characterize the behaviour of the TARM algorithm when the linear operator $\mathcal{A}$ is right-orthogonally invariant (ROIL). We show that the state evolution accurately predicts the performnace of the TARM algorithm. We also show that TARM converges much faster than other existing algorithms including SVP, NIHT, and RGrad. We further extend the TARM algorithm for matrix completion. We show that, although the state evolution cannot accurately predict the performance any more, the TARM algorithm with carefully tuned parameters still considerably outperforms its counterparts with comparable computational complexity.

\section{The TARM Algorithm}\label{Turbo_RARM}
\subsection{Algorithm Description}
In this section, we describe our proposed algorithm for affine rank minimization. The algorithm is inspired by the idea of turbo compressed sensing \cite{xue2017denoising, ma2014turbo}, hence the name turbo-type affine rank minimization (TARM).

The diagram of TARM is illustrated in Fig. \ref{tarm_fig}. There are two concatenated modules in TARM, namely, Module A and Module B. Module A estimates the low-rank matrix $\X^{\ast}$ via a linear estimator $\mathcal{E}(\cdot)$ based on the observation $\y$ and the input $\X$. Then the function $\mathcal{E}^{ext}(\cdot)$, which linearly combines $\mathcal{E}(\X)$ and $\X$, is employed to decorrelate the input and output estimation errors of Module A. The superscript ``ext" stands for extrinsic, since the output of $\mathcal{E}^{ext}(\cdot)$ is referred to as an extrinsic message. Module B has a similar structure as Module A does. In Module B, the output $\R$ of Module A is passed to a denoiser $\mathcal{D}(\cdot)$ which suppresses the estimation error by exploiting the low-rank structure of $\X^{\ast}$. The denoised output $\Z$ is then passed to a function $\mathcal{D}^{ext}(\cdot)$ which linearly combines $\Z$ and $\R$ for the decorrelation of input output errors. The two modules are executed iteratively to refine the estimates. Note that the TARM diagram in Fig. \ref{tarm_fig} is a matrix analogy of the Turbo-CS algorithm in Fig. 2 of \cite{xue2017denoising}.

The details of TARM are presented in Algorithm \ref{algorithm1}. We use index $t$ to denote the $t$-th iteration. There are three main steps at each iteration of TARM. The first step (Line 3 of Algorithm \ref{algorithm1}) is a linear estimation step which corresponds to Module A in Fig. \ref{tarm_fig}. This step combines $\mathcal{E}(\cdot)$ and $\mathcal{E}^{ext}(\cdot)$ since both are linear functions. Interestingly, this step can be interpreted as a gradient descent step for problem (5), where the parameter $\mu_t$ is the step size at the $t$-th iteration. In this sense, the TARM algorithm described here is closely related to the gradient-descent based SVP algorithm in \cite{jain2010guaranteed} and the NIHT algorithm in \cite{tanner2013normalized}. The second step (Line 4 of Algorithm \ref{algorithm1}) processes $\R^{(t)}$ with a denoiser $\mathcal{D}(\cdot)$. This step corresponds to the denoiser in Module B of TARM. There are various choices of $\mathcal{D}(\cdot)$ in the literature, such as the best rank-$r$ approximation \cite{eckart1936approximation} and the singular value thresholding (SVT) denoiser \cite{candes2009exact}. In this paper, we focus on the best rank-$r$ approximation defined by 
\begin{align}
	\mathcal{D}(\R)=\sum_{i=1}^r \sigma_{i}\bm{u}_i\bm{v}_i^T
\end{align}
where $\sigma_i, \bm{u}_i$, and $\bm{v}_i$ are respectively the $i$-th singular value and the corresponding left and right singular vectors of the input $\R$. We will show that, with this choice of $\mathcal{D}(\cdot)$, Module B allows an analytical characterization of its input output behavior. The third step (Line 5 of Algorithm \ref{algorithm1}) corresponds to function $\mathcal{D}^{ext}(\cdot)$ which is a linear combination of $\R^{(t)}$ and $\Z^{(t)}$ with the coefficients specified by $c_{t}$ and $\alpha_{t}$.

\begin{figure}[!ht]
	\centering
	\includegraphics[width=\linewidth]{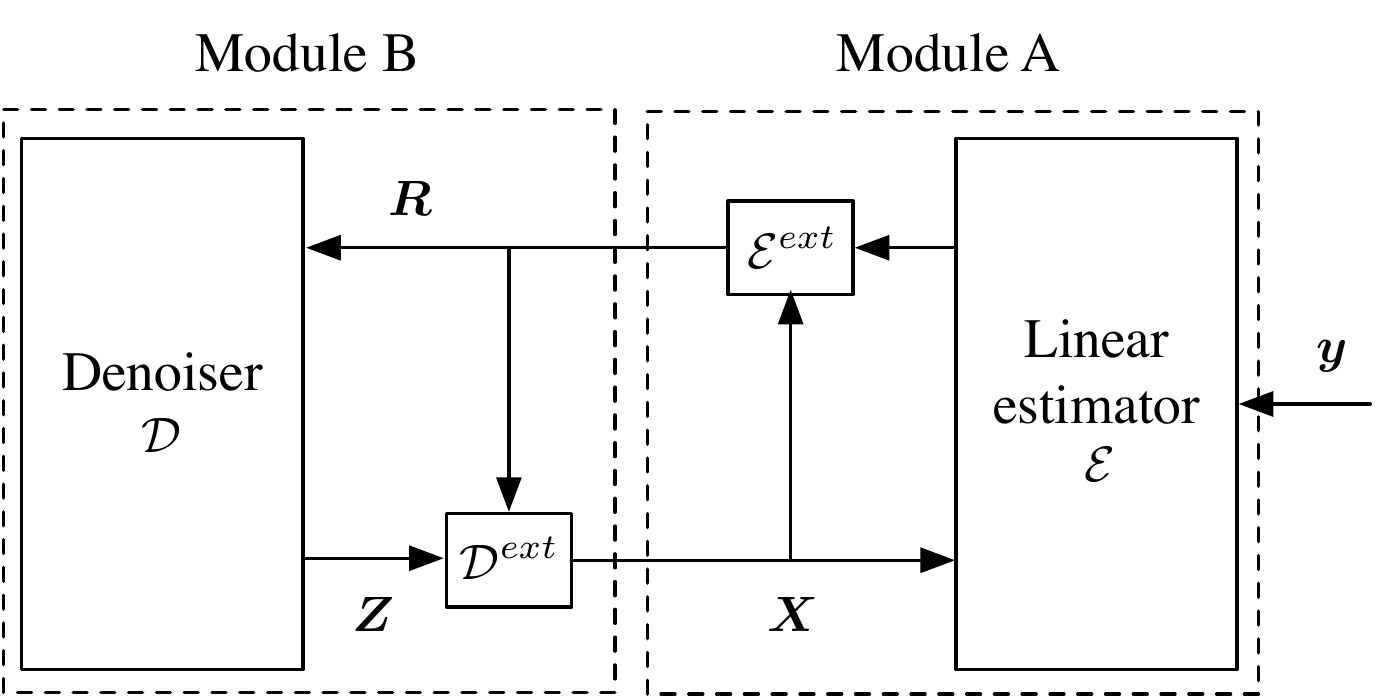}
	\caption{The diagram of the TARM algorithm.}\label{tarm_fig}
\end{figure}

\begin{algorithm}
\caption{TARM for affine rank minimization}\label{algorithm1}
\begin{algorithmic}[1]
\REQUIRE $\mathcal{A}, \y,\X^{(0)}=\bm{0}, t=0$
\\
\WHILE{the stopping criterion is not met}
\STATE $t=t+1$
\STATE $\R^{(t)}=\X^{(t-1)}+\mu_t\mathcal{A}^{T}(\y-\mathcal{A}(\X^{(t-1)}))$
\STATE $\Z^{(t)}=\mathcal{D}(\R^{(t)})$
\STATE $\X^{(t)}=c_{t}(\Z^{(t)}-\alpha_{t}\R^{(t)})$
\ENDWHILE
\ENSURE  $\Z^{(t)}$
\end{algorithmic}
\end{algorithm}

\subsection{Determining the Parameters of TARM}
In this section, we discuss how to determine the parameters $\{\mu_t\}$, $\{c_t\}$, and $\{\alpha_t\}$. We first note that when $c_t=1$ and $\alpha_{t}=0$ for any $t$, the algorithm reduces to the SVP or NIHT algorithm (depending on the choice of $\mu_t$). As such, the key difference of TARM from SVP and NITH resides in the choice of these parameters. By optimizing these parameters, the TARM algorithm aims to find a much steeper descendent direction for each iteration, so as to achieve a convergence rate much higher than SVP and NIHT. In specific, we follow the turbo principle \cite{xue2017denoising,ma2014turbo}, a special form of the more general message passing principle, to determine these parameters. That is, for each iteration $t$, $\{\mu_t, c_t, \alpha_t\}$ need to satisfy the following three conditions:
 \begin{itemize}
 \item Condition 1: 
\begin{align}
	\left<\R^{(t)}-\X^{\ast},\X^{(t-1)}-\X^{\ast}\right>=0,\label{cond1}
\end{align}
 \item Condition 2:
\begin{align}
	\left<\R^{(t)}-\X^{\ast},\X^{(t)}-\X^{\ast}\right>=0,\label{cond2}
\end{align}
\item Condition 3: For given $\X^{(t-1)}$,
\begin{align}
	\|\X^{(t)}-\X^{\ast}\|_F^2 \text{ is minimized under (\ref{cond1}) and (\ref{cond2}).}\label{cond3}
\end{align}
\end{itemize}
In the above, Condition 1 ensures that the input and output estimation errors of Module A are uncorrelated. Similarly, Condition 2 ensures that the input and output estimation errors of Module B are uncorrelated. Condition 3 ensures that the output estimation error of Module B is minimized over $\{\mu_t, c_t, \alpha_t\}$ for each iteration $t$. Note that in graphical-model based message passing, the out-going message on an edge is required to be independent of the incoming message on the edge. Since uncorrelatedness implies independence for Gaussian random variables, (\ref{cond1}) and (\ref{cond2}) can be seen as necessary conditions for Gaussian message passing. In this sense, the minimization in (\ref{cond3}) can be interpreted as finding the best estimate of $\X^{\ast}$ for each iteration under the Gaussian message passing framework.

We have the following lemma, with the proof given in Appendix A.

\begin{lemma}\label{lemma_p}
If Conditions 1-3 hold, then
\BS\label{parameters}
\begin{align}
	\mu_t &=\frac{\|\X^{(t-1)}-\X^{\ast}\|_F^2}{\left<\mathcal{A}(\X^{(t-1)}-\X^{\ast})-\n,\mathcal{A}(\X^{(t-1)}-\X^{\ast})\right>}\\
	\alpha_{t}&=\frac{-b_t\pm\sqrt{b_t^2-4a_t d_t}}{2a_t}\\
	c_{t}&=\frac{\left<\Z^{(t)}-\alpha_{t}\R^{(t)},\R^{(t)}\right>}{\|\Z^{(t)}-\alpha_{t}\R^{(t)}\|_F^2},\label{c_t}
\end{align}\ES
with 
\BS\label{abd}\begin{align}
	a_t &=  \|\R^{(t)}\|_F^2\|\R^{(t)}-\X^{\ast}\|_F^2\\
	b_t& = -\|\R^{(t)}\!\|_F^2\!\left<\R^{(t)}\!-\!\X^{\ast},\Z^{(t)}\right>\!-\!\|\Z^{(t)}\|_F^2\|\R^{(t)}\!-\!\X^{\ast}\|_F^2\notag\\&\ \ \  +\|\Z^{(t)}\!\|_F^2\!\left<\R^{(t)}\!-\X^{\ast},\X^{\ast}\right>\\
	d_t&=  \|\Z^{(t)}\|_F^2\left<\R^{(t)}-\X^{\ast},\Z^{(t)}-\X^{\ast}\right>.
\end{align}\ES
\end{lemma}

\begin{remark}
In (\ref{parameters}b), $\alpha_t$ has two possible choices and only one of them minimizes the error in (\ref{cond3}). From the discussion below (\ref{apda}), minimizing the square error in (\ref{cond3}) is equivalent to minimizing $\|\X^{(t)}-\R^{(t)}\|_F^2$. We have
\BS\label{2alpha}
\begin{align}
&\left\|\X^{(t)}-\R^{(t)}\right\|_F^2\notag\\
=&\left\|c_t(\Z^{(t)}-\alpha_{t}\R^{(t)})-\R^{(t)}\right\|_F^2\\
	=&\left\|\frac{\left<\Z^{(t)}\!-\!\alpha_{t}\R^{(t)},\R^{(t)}\right>}{\|\Z^{(t)}\!-\!\alpha_{t}\R^{(t)}\|_F^2}(\Z^{(t)}-\alpha_{t}\R^{(t)})-\R^{(t)}\right\|_F^2\\
		=&-\frac{\left<\Z^{(t)}-\alpha_{t}\R^{(t)},\R^{(t)}\right>^2}{\|\Z^{(t)}-\alpha_{t}\R^{(t)}\|_F^2}+\|\R^{(t)}\|_F^2
\end{align}\ES
where (\ref{2alpha}a) follows from substituting $\X^{(t)}$ in Line 5 of Algorithm \ref{algorithm1}, and (\ref{2alpha}b) follows by substituting $c_t$ in (\ref{parameters}c). Since $\|\R^{(t)}\|_F^2$ is invariant to $\alpha_t$, minimizing $\|\X^{(t)}-\R^{(t)}\|_F^2$ is equivalent to maximizing $\frac{\left<\Z^{(t)}-\alpha_{t}\R^{(t)},\R^{(t)}\right>^2}{\|\Z^{(t)}-\alpha_{t}\R^{(t)}\|_F^2}$. We choose $\alpha_t$ that gives a larger value of $\frac{\left<\Z^{(t)}-\alpha_{t}\R^{(t)},\R^{(t)}\right>^2}{\|\Z^{(t)}-\alpha_{t}\R^{(t)}\|_F^2}$.
\end{remark}

\begin{remark}
Similarly to SVP \cite{jain2010guaranteed} and NIHT \cite{tanner2013normalized}, the convergence of TARM can be analyzed by assuming that the linear operator $\mathcal{A}$ satisfies the restricted isometry property (RIP). The convergence rate of TARM is much faster than those of NIHT and SVP (provided that $\{\alpha_{t}\}$ are sufficiently small). More detailed discussions are presented in Appendix B.
\end{remark}

We emphasize that the parameters ${\mu_t},{\alpha_t}$, and  ${c_{t}}$ in (\ref{parameters}) are actually difficult to evaluate since $\X^{\ast}$ and $\n$ are unknown. This means that Algorithm 1 cannot rely on (\ref{parameters}) to determine $\mu_t, \alpha_t$ and $c_t$. In the following, we focus on how to approximately evaluate these parameters to yield practical algorithms. Based on different choices of the linear operator $\mathcal{A}$, our discussions are divided into two parts, namely, low-rank matrix recovery and matrix completion.

\section{Low-rank Matrix Recovery}
\subsection{Preliminaries}
In this section, we consider recovering $\X^{\ast}$ from measurement in (5) when the linear operator $\mathcal{A}$ is right-orthogonally invariant. Denote the vector form of $\X$ by $\x=\text{vec}(\X)=[\x_1^T,\x_2^T,\cdots,\x_n^T]^T$, where $\x_i$ is the $i$th column of $\X$. The linear operator $\mathcal{A}$ can be generally expressed as $\mathcal{A}(\X)=\A \mr{vec}(\X)=\A\x$ where $\A\in \mathbb{R}^{m\times n}$ is a matrix representation of $\mathcal{A}$. The adjoint operator $\mathcal{A}^{T}: \mathbb{R}^{m}\rightarrow \mathbb{R}^{n_1\times n_2}$ is defined by the transpose of $\A$ with $\x'=\mr{vec}(\X')=\mr{vec}(\mathcal{A}^{T}(\y'))=\A^{T}\y'$.
\begin{definition}
	Consider a linear operator $\mathcal{A}$ with matrix form $\A$, the SVD of $\A$ is $\A=\U_{A}\bm{\Sigma}_{A}\V_{A}^T$, where $\U_{A}$ and $\V_{A}$ are orthogonal matrices and $\bm{\Sigma_{A}}$ is a diagonal matrix. If $\V_{A}$ is a Haar distributed random matrix independent of $\bm{\Sigma}_{A}$, we say that $\mathcal{A}$ is a right-orthogonally invariant linear (ROIL) operator.
\end{definition}
We focus on two types of ROIL operators: partial orthogonal ROIL operators where the matrix form of $\mathcal{A}$ satisfies $\A\A^T=\I$, and Gaussian ROIL operators where the elements of $\A$ are i.i.d. Gaussian with zero mean.  For convenience of discussion, the linear operator $\mathcal{A}$ is normalized such that the length of each row of $\A$ is 1. It is worth noting that from the perspective of the algorithm, $\A$ is deterministic since $\A$ is known by the algorithm. However, the randomness of $\A$ has impact on parameter design and performance analysis, as detailed in what follows.

\subsection{Parameter Design}
We now determine the parameters in (\ref{parameters}) when ROIL operators are involved. We show that (\ref{parameters}) can be approximately evaluated without the knowledge of $\X^{\ast}$. Since $\{c_t\}$ in (\ref{parameters}c) can be readily computed given $\{\alpha_t\}$, we focus on the calculation of $\{\mu_t\}$ and $\{\alpha_t\}$.

We start with $\mu_t$. From (\ref{parameters}a), we have
\BS\label{mu}
\begin{align}
		\mu_{t}&=\frac{\|\X^{(t-1)}-\X^{\ast}\|_F^2}{\left<\mathcal{A}(\X^{(t-1)}-\X^{\ast})-\n,\mathcal{A}(\X^{(t-1)}-\X^{\ast})\right>}\\
		&\approx\frac{\|\X^{(t-1)}-\X^{\ast}\|_F^2}{\|\mathcal{A}(\X^{(t-1)}-\X^{\ast})\|_2^2}\\
		&=\frac{1}{\|\mathcal{A}(\frac{\X^{(t-1)}-\X^{\ast}}{\|\X^{(t-1)}-\X^{\ast}\|_F})\|_2^2}\\
		&=\frac{1}{\tilde{\x}^T\V_A\bm{\Sigma}_A^T\bm{\Sigma}_A\V_A^T\tilde{\x}}\\
		&=\frac{1}{\v_{A}^T\bm{\Sigma}_{A}^T\bm{\Sigma}_A\v_{A}}\\
		&\approx\frac{n}{m}
\end{align}\ES
where (\ref{mu}b) follows from $\left<\n,\mathcal{A}(\X^{(t-1)}-\X^{\ast})\right>\approx0$, (\ref{mu}d) follows by utilizing the matrix form of $\mathcal{A}$ and $\tilde{\x}=\frac{\text{vec}(\X^{(t-1)}-\X^{\ast})}{\|\X^{(t-1)}-\X^{\ast}\|_F}$, and (\ref{mu}e) follows by letting $\v_{A}=\V_A^T\tilde{\x}$. As $\mathcal{A}$ is a ROIL operator, $\V_A$ is haar distributed and is approximately independent of $\tilde{\x}$, implying that $\v_{A}$ is a unit vector uniformly distributed over the sphere $\|\v_{A}\|_2=1$. Then, the last step of (\ref{mu}) follows by noting $\text{Tr}(\bm{\Sigma}_A^T\bm{\Sigma}_A)=m$.

We next consider the approximation of $\alpha_t$. We first note
\BS\label{alpha_t}
\begin{align}
	&\ \ \ \left<\R^{(t)}-\X^{\ast},\X^{(t)}-\X^{\ast}\right>
	\\&=\left<\R^{(t)}-\X^{\ast},c_{t}(\Z^{(t)}-\alpha_{t}\R^{(t)})-\X^{\ast}\right>\\
	&\approx  c_t\left<\R^{(t)}-\X^{\ast},\Z^{(t)}-\alpha_{t}\R^{(t)}\right>
\end{align}\ES
where (\ref{alpha_t}a) follows by substituting $\X^{(t)}$ in line 5 of Algorithm 1, and (\ref{alpha_t}b) follows from $\left<\R^{(t)}-\X^{\ast},\X^{\ast}\right>\approx 0$ implying that the error $\R^{(t)}-\X^{\ast}$ is uncorrelated with the original signal $\X^{\ast}$. Combining (\ref{alpha_t}) and Condition 2 in (\ref{cond2}), we have 

\BS\label{alpha_t1}
\begin{align}
	\alpha_t &=\frac{\left<\R^{(t)}-\X^{\ast},\Z^{(t)}\right>}{\left<\R^{(t)}-\X^{\ast},\R^{(t)}\right>}\\&\approx\frac{\left<\R^{(t)}-\X^{\ast},\mathcal{D}(\R^{(t)})\right>}{\left<\R^{(t)}-\X^{\ast},\R^{(t)}-\X^{\ast}\right>}\\
	&\approx\frac{\left<\R^{(t)}-\X^{\ast},\mathcal{D}(\R^{(t)})\right>}{nv_t}\\
	&\approx\frac{1}{n}\sum_{i,j}\frac{\partial \mathcal{D}(\R^{(t)})}{\partial R_{i,j}^{(t)}}\\
	&=\frac{1}{n}\mr{div}(\mathcal{D}(\R^{(t)}))
\end{align}\ES
where (\ref{alpha_t1}b) follows from $\Z^{(t)}=\mathcal{D}(\R^{(t)})$ and $\left<\R^{(t)}-\X^{\ast},\X^{\ast} \right >\approx0$, (\ref{alpha_t1}c) follows from the Gaussian approximation that the elements of $\R^{(t)}-\X^{\ast}$ are i.i.d. Gaussian with zero mean and variance $v_t$, (\ref{alpha_t1}d) follows from Stein's lemma \cite{stein1981estimation} since we approximate the entries of $\R^{(t)}-\X^{\ast}$ as i.i.d. Gaussian distributed, and (\ref{alpha_t1}e) is from the definition of the divergence $\text{div}(\cdot)$.
 
\subsection{State Evolution}
We now characterize the performance of TARM for low-rank matrix recovery. Recall that some assumptions are involved in determining the algorithm parameters in the preceding subsection. We formally present these assumptions as follows.

\begin{assumption}
For each iteration $t$, the orthogonal matrix $\V_A$ is independent of Module A's input estimation error $\X^{(t-1)}-\X^{\ast}$.
\end{assumption}

\begin{assumption}
For each iteration $t$, the output error of Module A, given by $\R^{(t)}-\X^{\ast}$, resembles an i.i.d. Gaussian noise, i.e., the elements of $\R^{(t)}-\X^{\ast}$ are independently and identically drawn from $\mathcal{N}(0, v_t)$, where $v_t$ is the output variance of Module A at iteration $t$.
\end{assumption}

The above two assumptions will be verified by the numerical results presented in the next subsection. Similar assumptions have been introduced in the design of Turbo-CS in \cite{ma2014turbo} (see also \cite{ma2017orthogonal}). Later, these assumptions were rigorously analyzed in \cite{rangan2017vector,takeuchi2017rigorous} using the conditioning technique \cite{bayati2011dynamics}. Based on that, state evolution was established to characterize the behavior of the Turbo-CS algorithm. However, the analyses in \cite{rangan2017vector,takeuchi2017rigorous} are focused on the case that the denoiser $\mathcal{D}(\cdot)$ is separable, i.e., the function $\mathcal{D}(\cdot)$ is individually applied to each element of the input, while the denoisers involved here (such as the best-rank-$r$ appriximation in (6)) are non-separable. Therefore, the technique in \cite{rangan2017vector}-\cite{bayati2011dynamics} cannot be directly applied here. The recent work \cite{berthier2017state} on state evolution of AMP for non-seperable denoisers may shed some light on a possible rigorous justification of the assumptions. In this paper, we establish state evolution based on the two assumptions. We leave the rigorous proof (without imposing the assumptions) as future work.

Assumptions 1 and 2 allow to decouple Module A and Module B in the analysis of the TARM algorithm. We derive two mean square error (MSE) transfer functions, one for each module, to characterize the behavior of the TARM algorithm.

We first consider the MSE behavior of Module A. Denote the output MSE of Module A at iteration $t$ by
\begin{align}
	MSE_A^{(t)}=\frac{1}{n}\|\R^{(t)}-\X^{\ast}\|_F^2.
\end{align}
The following theorem gives the asymptotic MSE of Module A when the dimension of $\X^{\ast}$ goes to infinity, with the proof given in Appendix C.

\begin{theorem}
	Assume that Assumption 1 holds, and let $\mu=\frac{n}{m}$. Then,
	\begin{align}
		MSE_A^{(t)}\xrightarrow[]{\text{a.s.}}f(\tau_t)
	\end{align}
as $m,n\rightarrow\infty$ with $\frac{m}{n}\rightarrow \delta$, where $\frac{1}{n}\|\X^{(t-1)}-\X^{\ast}\|_F^2\rightarrow \tau_t$ as $n\rightarrow\infty$.
For partial orthogonal ROIL operator $\mathcal{A}$,
\BS\label{theorem1}
	\begin{align}
		f(\tau)= \left(\frac{1}{\delta}-1\right)\tau+\sigma^2
	\end{align}
	and for Gaussian ROIL operator $\mathcal{A}$,
	\begin{align}
		f(\tau)= \frac{1}{\delta}\tau+\sigma^2.
	\end{align}
	\ES
\end{theorem}

We now consider the MSE behavior of Module B. We start with the following useful lemma, with the proof given in Appendix D.
\begin{lemma}
	Assume that $\R^{(t)}$ satisfies Assumption 2, $\|\X^{\ast}\|_F^2=n$, and the empirical distribution of eigenvalue $\theta$ of $\frac{1}{n_2}\X^{\ast T}\X^{\ast}$ converges almost surely to the density function $p(\theta)$ as $n_1,n_2,r\rightarrow \infty$ with $\frac{n_1}{n_2}\rightarrow\rho,\frac{r}{n_2}\rightarrow\lambda$. Then,
\BS\label{ac}
	\begin{align}
		\alpha_{t}&\xrightarrow[]{\text{a.s.}}\alpha(v_t)\\
		c_{t}&\xrightarrow[]{\text{a.s.}}c(v_t)
	\end{align}\ES
as $n_1,n_2,r\rightarrow \infty$ with $\frac{n_1}{n_2}\rightarrow\rho,\frac{r}{n_2}\rightarrow\lambda$, where 
\BS\label{alpha_c} 
\begin{align}
	\alpha(v)&=\left|1\!-\!\frac{1}{\rho}\right|\lambda\!+\!\frac{1}{\rho}\lambda^2\!+\!\!2\left(\min\left(\!1,\frac{1}{\rho}\!\right)\!-\!\frac{\lambda}{\rho}\right)\!\lambda\Delta_1(v)\\
	c(v)&=\frac{1+\lambda(1+\frac{1}{\rho})v+\lambda  v^2\Delta_2-\alpha(v)(1+v)}{(1\!-\!2\alpha(v))(1\!+\!\lambda(1\!+\!\frac{1}{\rho})v+\lambda v^2\Delta_2)+\alpha(v)^2(1\!+\!v)}
\end{align}\ES
with $\Delta_1$ and $\Delta_2$ defined by
\BS\label{delta}
\begin{align}
	\Delta_1(v)&=\int_0^{\infty}\frac{(v+\theta^2)(\rho v+\theta^2)}{(\sqrt{\rho}v-\theta^2)^2} p(\theta)d\theta\\
	\Delta_2&=\int_{0}^{\infty} \frac{1}{\theta^2} p(\theta)d \theta.
\end{align}\ES
\end{lemma}

Denote the output MSE of Module B at iteration $t$ by
\begin{align}
	MSE_B^{(t)}=\frac{1}{n}\|\X^{(t)}-\X^{\ast}\|_F^2.
\end{align}
The output MSE of Module B is characterized by the following theorem.

\begin{theorem}
Assume that Assumption 2 holds, and let $\|\X^{\ast}\|_F^2=n$. Then, the output MSE of Module B 
	\begin{align}\label{mseb}
		MSE_B^{(t)}\!\xrightarrow[]{\text{a.s.}}\!g(v_t)\!\triangleq\! \frac{v_t-\lambda\left(1+\frac{1}{\rho}\right) v_t-\lambda v_t^2\Delta_2}{\frac{v_t-\lambda\left(1+\frac{1}{\rho}\right) v_t-\lambda v_t^2\Delta_2}{1+\lambda\left(1+\frac{1}{\rho}\right)v_t+\lambda v_t^2\Delta_2}\alpha(v_t)^2\!+\!(1\!-\!\alpha(v_t))^2}\!-\!v_t
	\end{align}
	as $n_1,n_2,r\rightarrow \infty$ with $\frac{n_1}{n_2}\rightarrow\rho,\frac{r}{n_2}\rightarrow\lambda$, where $\alpha$ and $\Delta_2$ are given in Lemma 2, and $\frac{1}{n}\|\R^{(t)}-\X^{\ast}\|_F^2\xrightarrow[]{\text{a.s.}}v_t$.
\end{theorem}

\begin{remark}
$\Delta_1$ and $\Delta_2$ in (\ref{delta}) may be difficult to obtain since $p(\theta)$ is usually unknown in practical scenarios. We now introduce an approximate MSE expression that does not depend on $p(\theta)$:
\begin{align}\label{uppermseb}
	g(v_t)\approx \bar{g}(v_t)\triangleq \frac{v_t-\lambda(1+\frac{1}{\rho}) v_t}{(1-\alpha)^2}-v_t
\end{align}
where $\alpha=\alpha(0)=\left|1-\frac{1}{\rho}\right|\lambda-\frac{1}{\rho}\lambda^2+2\min\left(1,\frac{1}{\rho}\right)\lambda$. Compared with $g(v_t)$, $\bar{g}(v_t)$ omits two terms $-\lambda v_t^2\Delta$ and $\frac{v_t-\lambda\left(1+\frac{1}{\rho}\right) v_t-\lambda v_t^2\Delta}{1+\lambda\left(1+\frac{1}{\rho}\right)v_t+\lambda v_t^2\Delta}\alpha(v_t)^2$ and replaces $\alpha(v_t)$ by $\alpha$. Recall that $v_t$ is the mean square error at the $t$-iteration. As the iteration proceeds, we have $v_t\ll 1$, and hence $g(v_t)$ can be well approximated by $\bar{g}(v_t)$, as seen later from Fig. \ref{fig_state}.
\end{remark}

Combining Theorems 1 and 2, we can characterize the MSE evolution of TARM by
\BS\label{statee}
\begin{align}
	v_t&=f(\tau_t)\\
	\tau_{t+1} &=g(v_{t}).
\end{align}\ES
The fixed point of TARM's MSE evolution in (\ref{statee}) is given by
\begin{align}\label{sefunction}
	\tau^{\ast}=g(f(\tau^{\ast})).
\end{align}
The above fixed point equation can be used to analysis the phase transition curves of the TARM algorithm. It is clear that the fixed point $\tau^{\ast}$ of (\ref{sefunction}) is a function of $\{\delta,\rho,\lambda,\Delta,\sigma\}$. For any given $\{\delta,\rho,\lambda,\Delta,\sigma\}$, we say that the TARM algorithm is successful if the corresponding $\tau^{\ast}$ is below a certain predetermined threshold. The critical values of $\{\delta,\rho,\lambda,\Delta,\sigma\}$ define the phase transition curves of the TARM algorithm.

\subsection{Numerical Results}
Some simulation settings are as follows. For the case of partial orthogonal ROIL operators, we generate a partial orthogonal ROIL operator with the matrix form $\A=\S\W$, where $\S\in \mathbb{R}^{m\times n}$ is a random perturbation matrix and $\W\in \mathbb{R}^{n\times n}$ is a discrete cosine transform (DCT) matrix. For the case of Gaussian ROIL operators, we generate an i.i.d. Gaussian random matrix of size $m\times n$ with elements drawn from $\mathcal{N}(0,\frac{1}{n})$. The rank-$r$ matrix $\X^{\ast}\in \mathbb{R}^{n_1\times n_2}$ is generated by the product of two i.i.d. Gaussian matrices of size $n_1\times r$ and $r\times n_2$.

\subsubsection{Verification of the assumptions}
We first verify Assumption 1 using Table \ref{table1}. Recall that if Assumption 1 holds, the approximations in the calculation of $\mu_t$ in (\ref{mu}) become accurate. Thus, we verify Assumption 1 by comparing the value of $\mu_t$ calculated by (\ref{parameters}a) with $\frac{n}{m}$ by (\ref{mu}). We record the $\mu_t$ of the first 8 iterations of TARM in Table \ref{table1} for low-rank matrix recovery with a partial orthogonal ROIL operator. As shown in Table \ref{table1}, the approximation $\mu_t=\frac{n}{m}$ is close to the real value calculated by (\ref{parameters}a) which verifies Assumption 1. We then verify Assumption 2 using Fig. \ref{qqplots}, where we plot the QQplots of the input estimation errors of Module A with partial orthogonal and Gaussian ROIL operators. The QQplots show that the output errors of Module A closely follow a Gaussian distribution, which agrees with Assumption 2.

\begin{table*}[!h]
\centering
\begin{tabular}{|c|c|c|c|c|c|c|c|c|}
\hline
 iteration $t$  & 1 & 2 & 3 & 4 & 5& 6&  7&  8\\ \hline
 $\frac{n}{m}=2.5$ &2.4960 & 2.4988  &  2.4944 &   2.4948& 2.4938 &2.4950  & 2.4976 & 2.4968 \\ \hline
$\frac{n}{m}=3.3333$ &  3.3283 & 3.3273  &  3.3259 & 3.3268  & 3.3295 & 3.3267 &  3.3299& 3.3269 \\ \hline
$\frac{n}{m}=5$ & 4.9994  &  4.9998 &  5.0034 &4.9995   & 5.0005 & 5.0058 & 5.0078 & 5.0011 \\ \hline
\end{tabular}
\caption{$\mu_t$ calculated by (10a) for the 1st to 8th iterations of one random realization of the algorithm with a partial orthogonal ROIL operator. $n_1=n_2=1000$, $r=30$, $\sigma=10^{-5}$.}\label{table1}
\end{table*}

\begin{figure}
	\centering
	\includegraphics[width=0.5\linewidth]{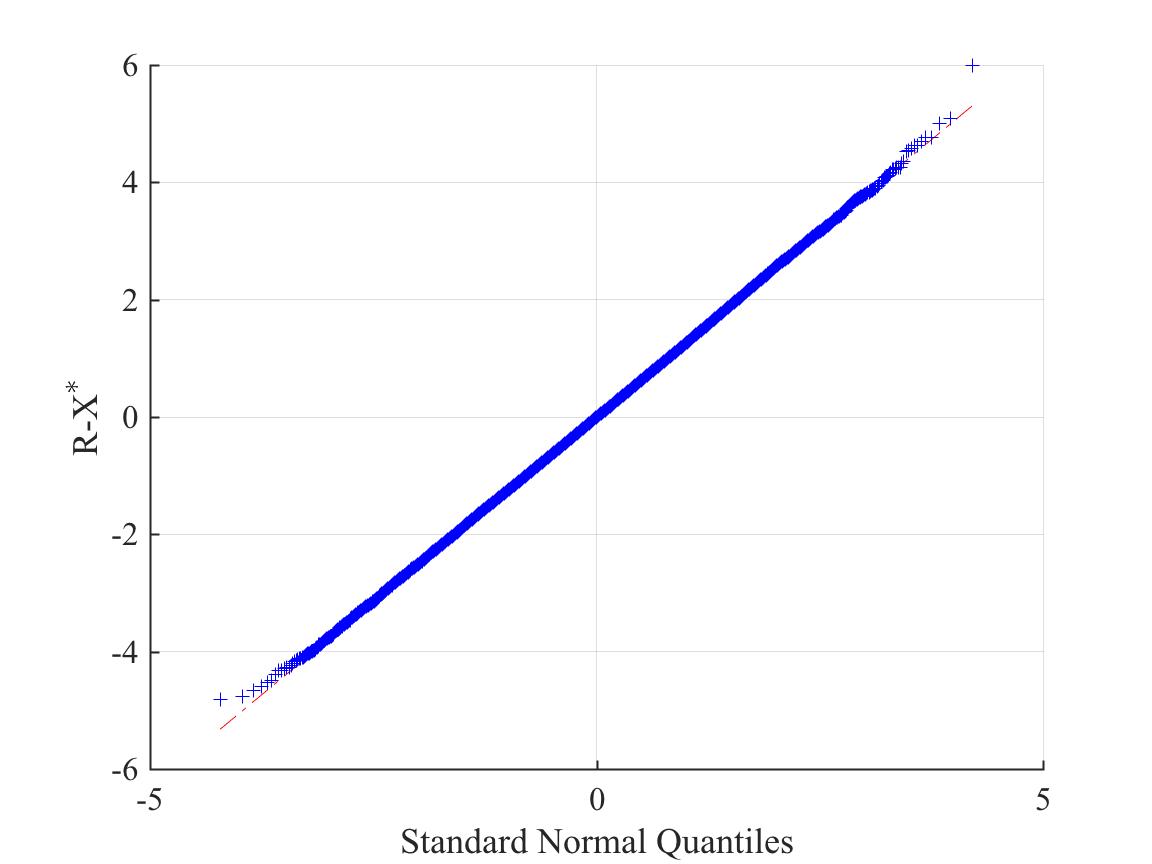}\includegraphics[width=0.5\linewidth]{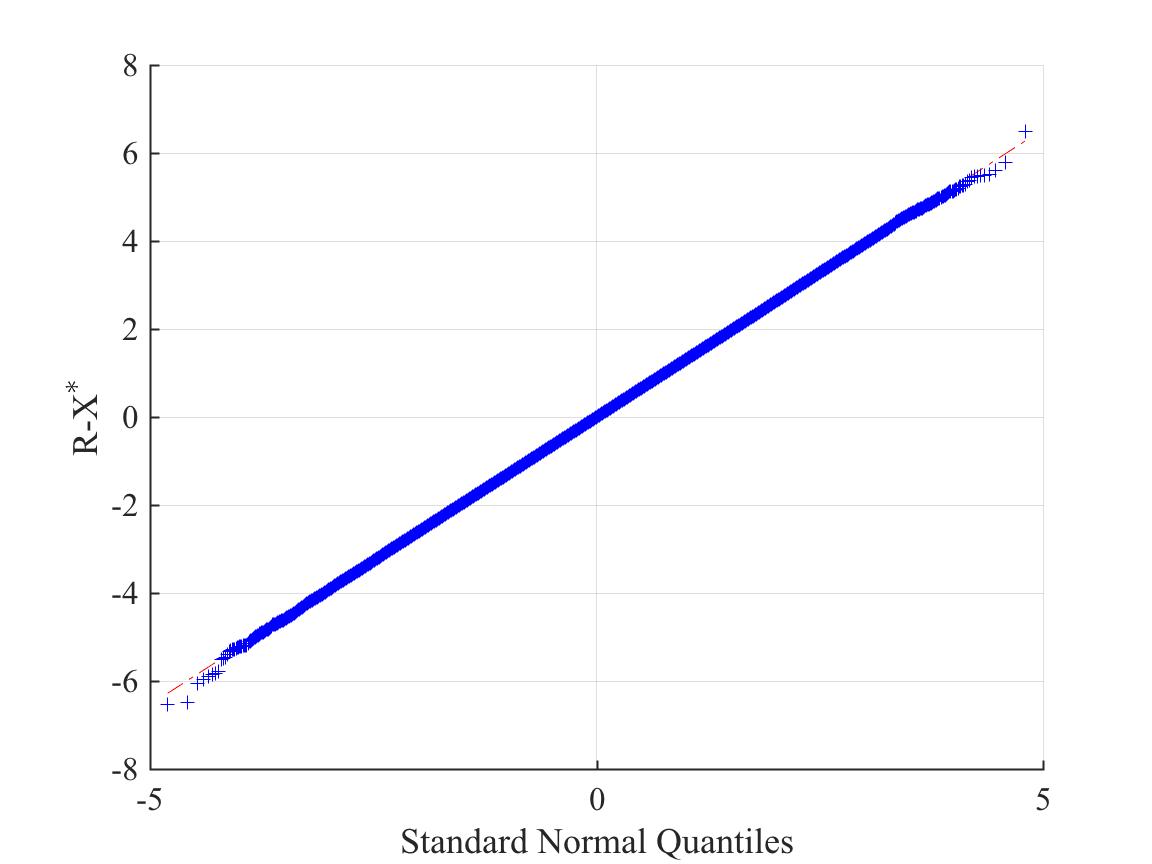}
	\caption{The QQplots of the output error of Module B in the 2nd iteration of TARM. Left: $\mathcal{A}$ is a Gaussian ROIL operator. Right: $\mathcal{A}$ is a partial orthogonal ROIL operator. Simulation settings: $n_1=100,n_2=120, \frac{m}{n_1 n_2}=0.3, \frac{r}{n_2}=0.25, \sigma^2=0$.}\label{qqplots}
\end{figure}

\subsubsection{State evolution}
We now verify the state evolution of TARM given in (\ref{statee}). We plot the simulation performance of TARM and the predicted performance by the state evolution in Fig. \ref{fig_state}. From the two subfigures in Fig. \ref{fig_state}, we see that the state evolution of TARM is accurate when the dimension of $\X^{\ast}$ is large enough for both partial orthogonal and Gaussian ROIL operators. We also see that the state evolution with $g(\cdot)$ replaced by the approximation in (\ref{uppermseb}) (referred to as "Approximation" in Fig. \ref{fig_state}) provides reasonably accurate performance predictions. This makes the upper bound very useful since it does not require the knowledge of the singular value distribution of $\X^{\ast}$.

\begin{figure}[!ht]
	\centering
	\includegraphics[width=0.5\linewidth]{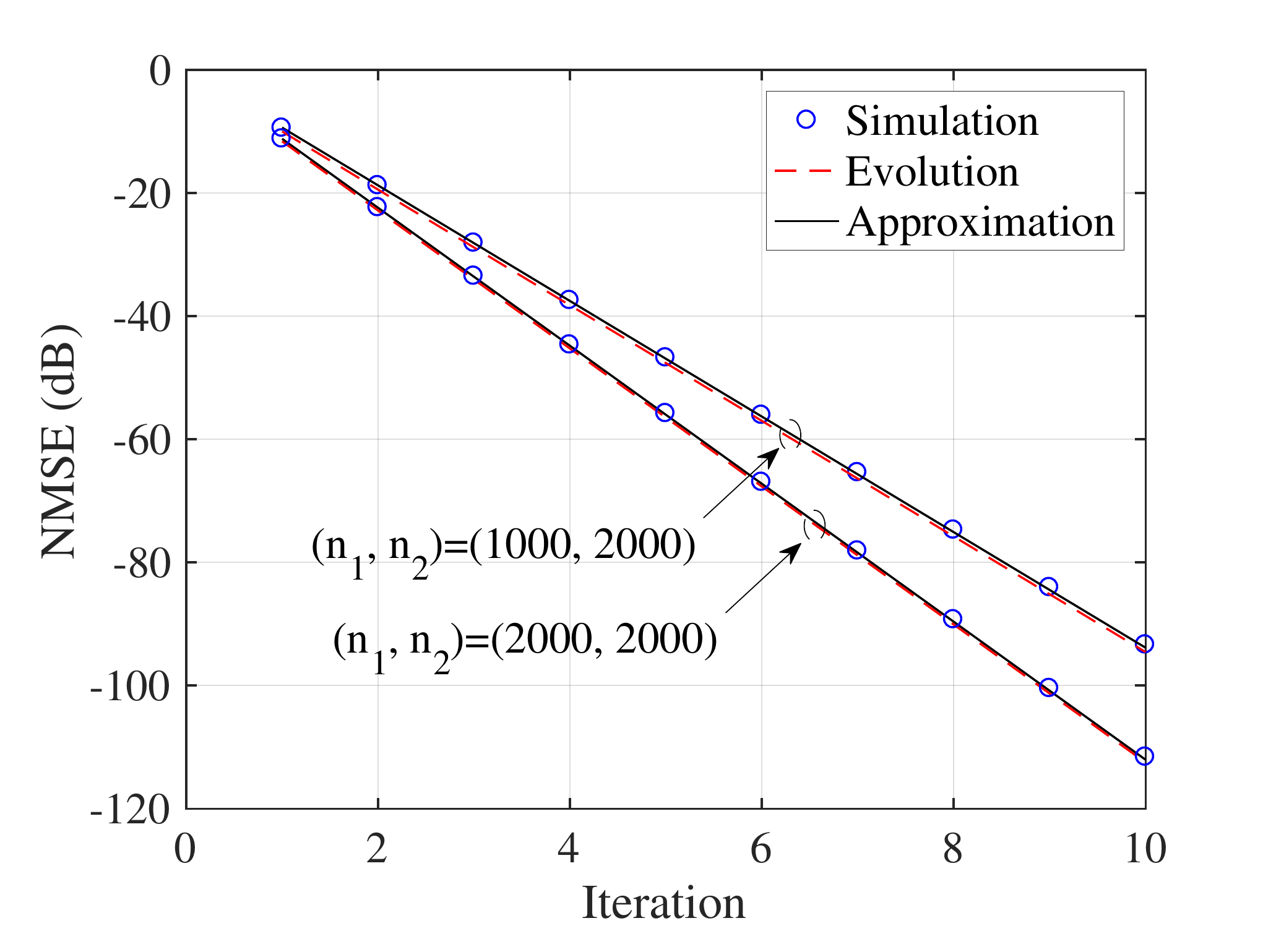}\includegraphics[width=0.5\linewidth]{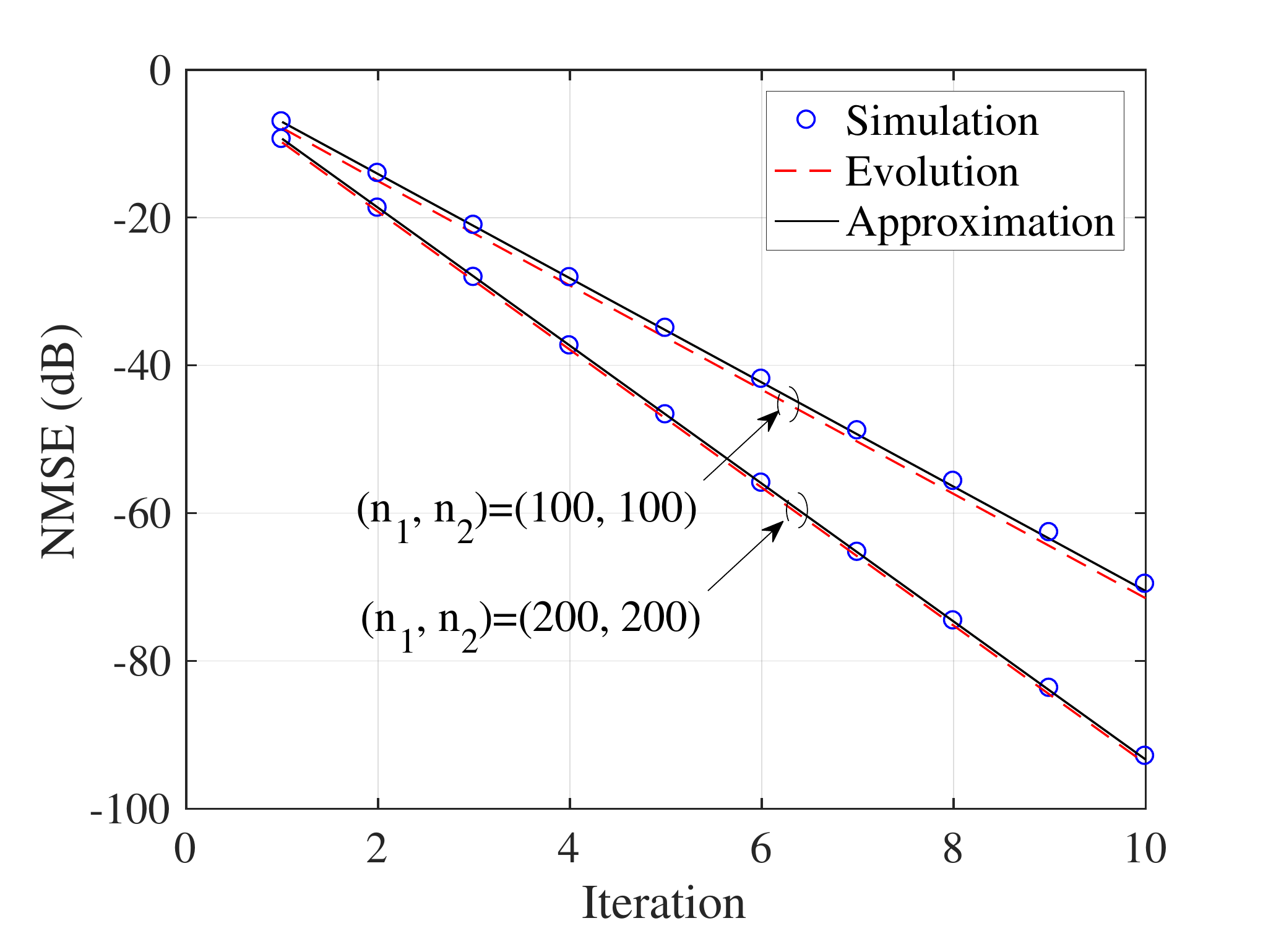}
	\caption{Left: State evolution of TARM for partial orthogonal ROIL operator. $r=40, m/n=m/(n_1 n_2)=0.35, \sigma^2=0$. The size of $\X^{\ast}$ is shown in the plot. Right: State evolution of TARM for Gaussian ROIL operator. $r=4, m/n=0.35$, $\sigma^2=0$. The size of $\X^{\ast}$ is shown in the plot.}\label{fig_state}
\end{figure}
\begin{figure}[!ht]
\centering
	\includegraphics[width=0.5\linewidth]{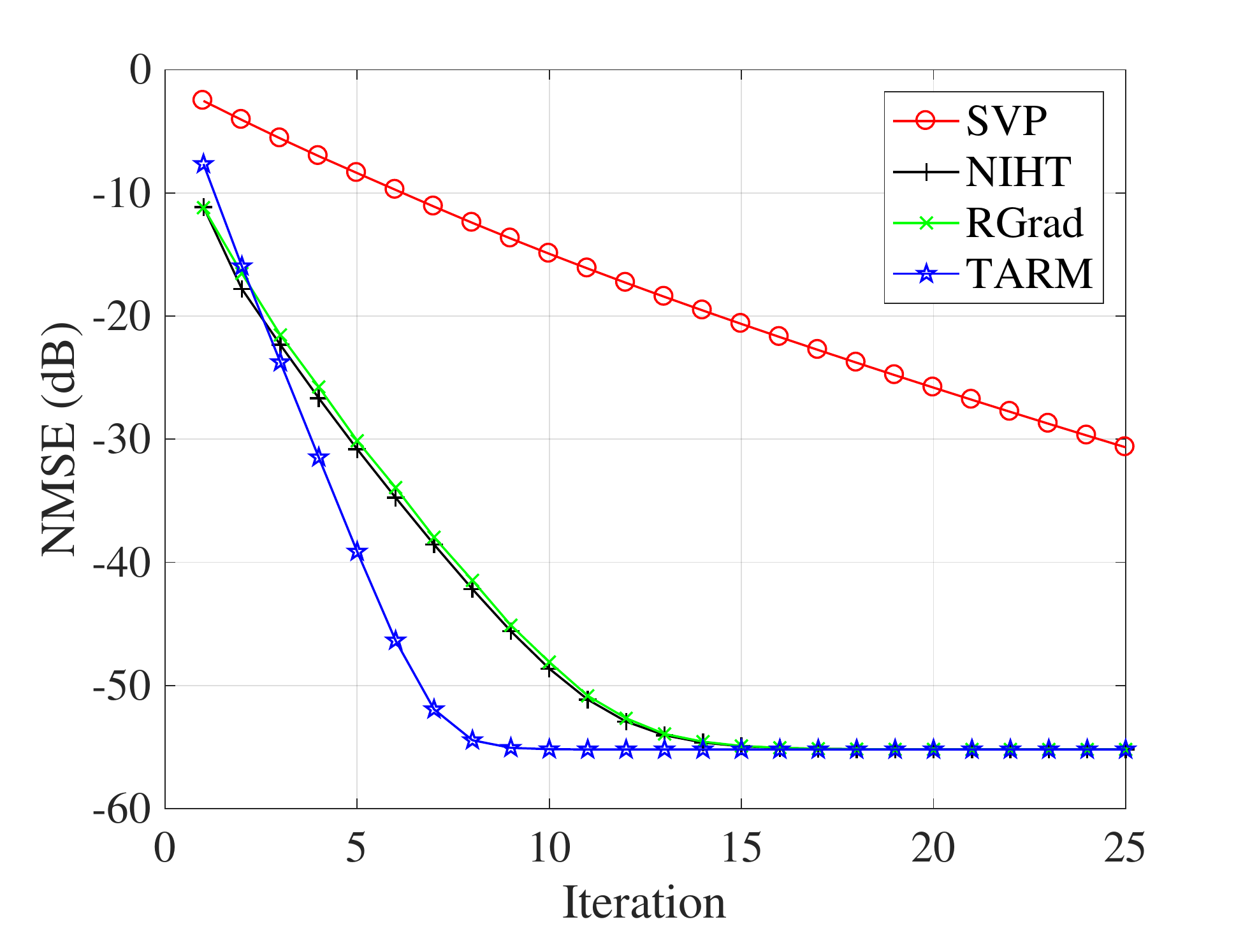}\includegraphics[width=0.5\linewidth]{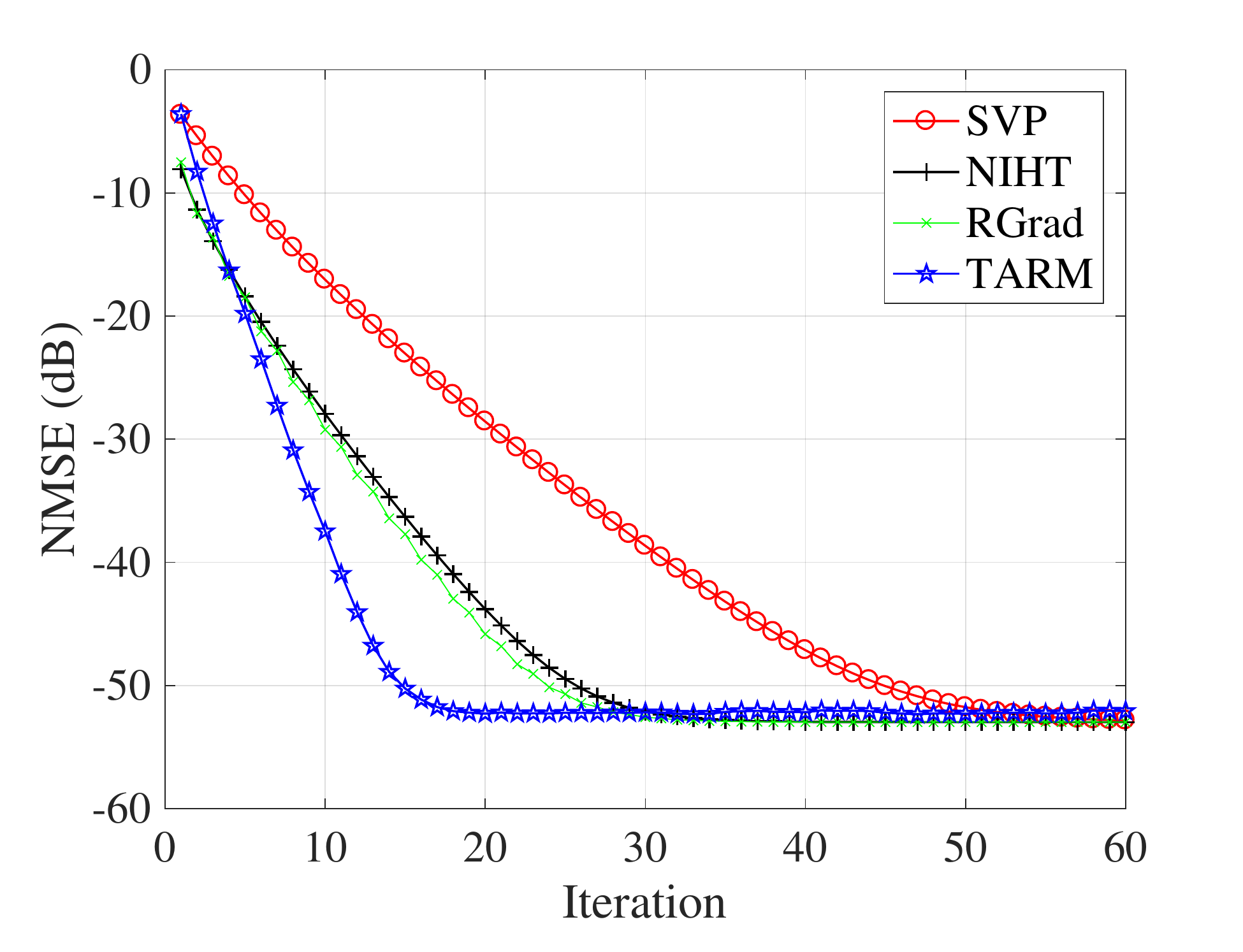}
	\includegraphics[width=0.5\linewidth]{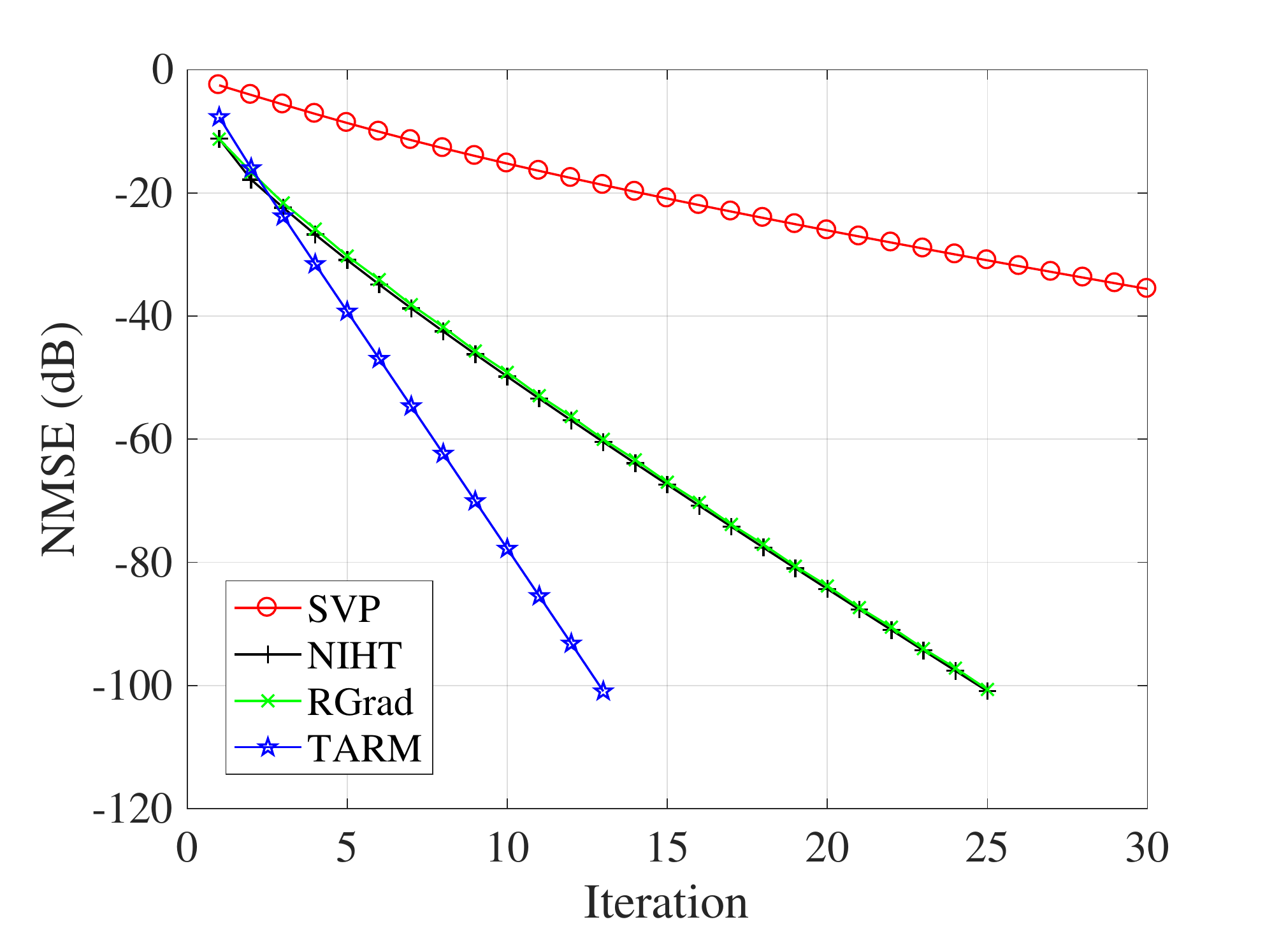}\includegraphics[width=0.5\linewidth]{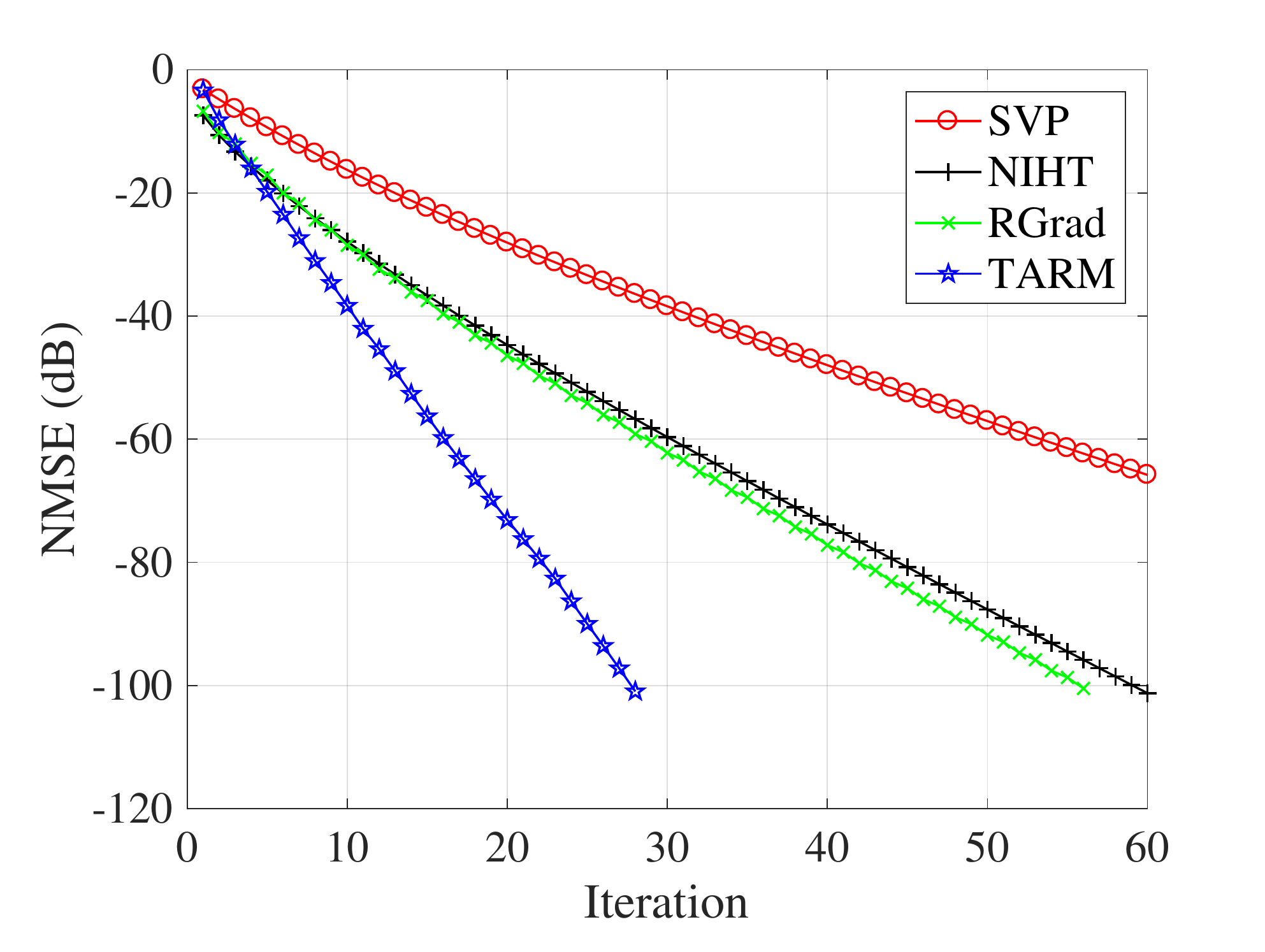}
	\caption{Comparison of algorithms. Top left: $\mathcal{A}$ is a partial orthogonal ROIL operator with $n_1=n_2=1000$, $r=50,m/n=0.39, \sigma^2=10^{-5}$. Top right: $\mathcal{A}$ is a random Gaussian ROIL operator with $n_1=n_2=80$, $r=10, p=(n_1+n_2-r)\times r,m/p=3, \sigma^2=10^{-5}$. Bottom left: $\mathcal{A}$ is a random Gaussian ROIL operator with $n_1=n_2=80$, $r=10, p=(n_1+n_2-r)\times r,m/p=3, \sigma^2=0$. Bottom right: $\mathcal{A}$ is a random Gaussian ROIL operator with $n_1=n_2=80$, $r=10, p=(n_1+n_2-r)\times r,m/p=3, \sigma^2=0$.}\label{plot_po}
\end{figure}

\subsubsection{Performance comparisons}
We compare TARM with the existing state-of-the-art algorithms for low-rank matrix recovery with partial orthogonal and Gaussian ROIL operators. The following algorithms are involved: singular value projection (SVP) \cite{jain2010guaranteed}, normalized iterative 
hard thresholding \cite{tanner2013normalized}, and  Riemannian gradient descent (RGrad) \cite{wei2016guarantees}. We compare these algorithms under the same settings for 100 times, and the final results are averaged over all the comparisons. We plot the per iteration normalized mean square error (NMSE) in Fig. \ref{plot_po}. From Fig. \ref{plot_po}, we see that TARM converges much faster than NIHT and RGrad for both Gaussian ROIL operators and partial orthogonal ROIL operators.

\subsubsection{Empirical phase transition}
The phase transition curve characterized the tradeoff between measurement rate $\delta$ and the largest rank $r$ that an algorithm succeeds in the recovery of $\X^{\ast}$. We consider an algorithm to be successful in recovering the low-rank matrix $\X^{\ast}$ when the following conditions are satisfied: 1) the normalized mean square error $\frac{\|\X^{(t)}-\X^{\ast}\|^2_F}{\|\X^{\ast}\|^2_F}\leq 10^{-6}$; 2) the iteration time $t<1000$. The dimension of the manifold of $n_1\times n_2$ matrices of rank $r$ is $r(n_1+n_2-r)$ \cite{vandereycken2013low}. Thus, for any algorithm, the minimal number of measurements for successful recovery is $r(n_1+n_2-r)$, i.e., $m\geq r(n_1+n_2-r)$. Then, an upper bound for successful recovery is $r\leq\frac{n_1+n_2-\sqrt{(n_1+n_2)^2-4m}}{2}$. In Fig. \ref{pt_orth}, we plot the phase transition curves of the algorithms mentioned before. From Fig. \ref{pt_orth}, we see that the phase transition curve of TARM is the closest to the upper bound and considerably higher than the curves of NIHT and RGrad.
\begin{figure}[!ht]
	\centering
	\includegraphics[width=\linewidth]{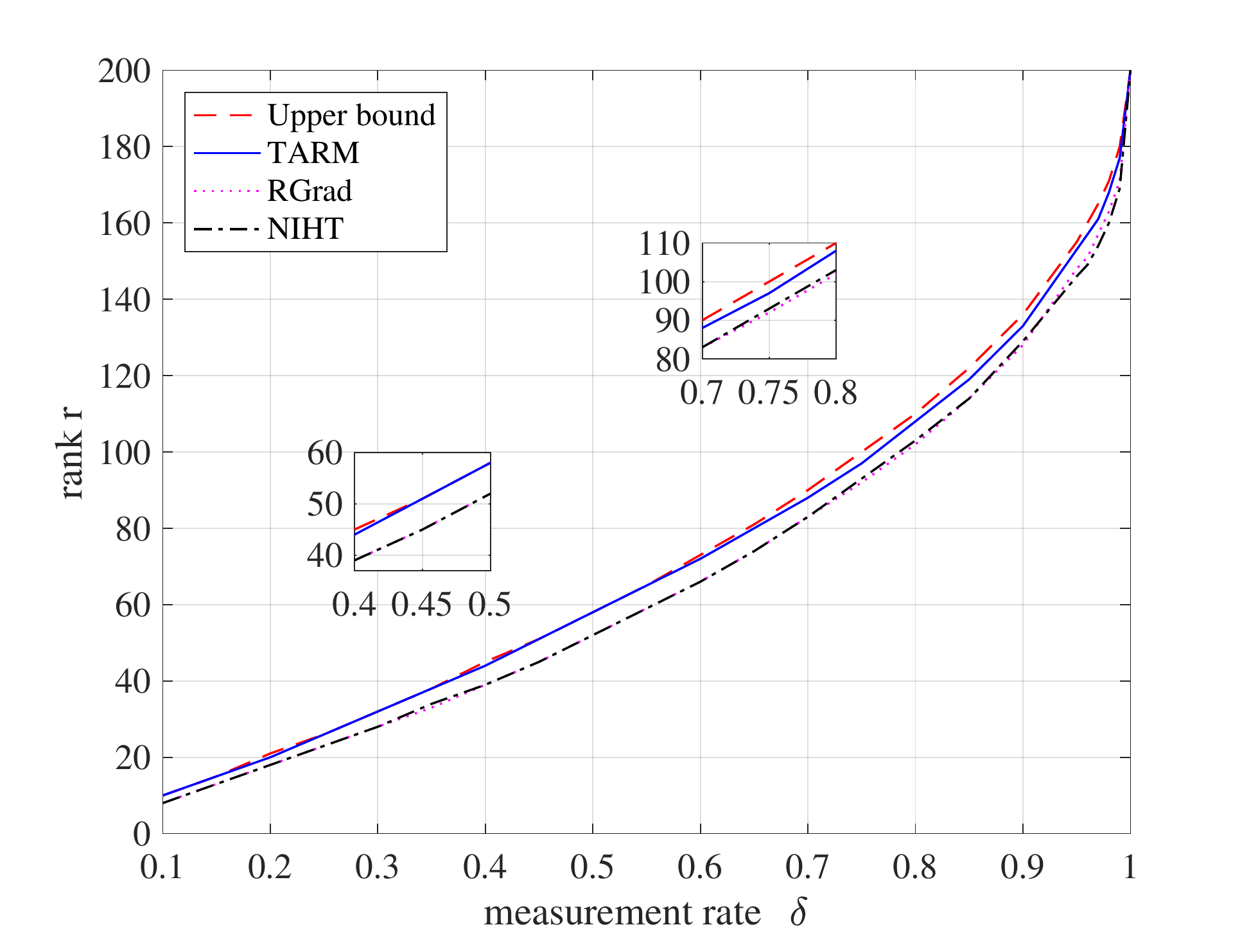}
	\caption{The phase transition curves of various low-rank matrix recovery algorithms with a partial orthogonal ROIL operator. $n_1=n_2=200$, $\sigma^2=0$. The region below each phase transition curve corresponds to the situation that the corresponding algorithm successfully recovers $\X^{\ast}$.}\label{pt_orth}
\end{figure}

\section{Matrix Completion}
In this section, we consider TARM for the matrix completion problem, where the linear operator $\mathcal{A}$ is a selector which selects a subset of the elements of the low-rank matrix $\X^{\ast}$. With such a choice of $\mathcal{A}$, the two assumptions in Section III for low-rank matrix recovery do not hold any more; see, e.g., Fig. \ref{qqplot_mc}. Thus, $\mu_{t}$ given in (\ref{mu}) and $\alpha_{t}$ in (\ref{alpha_t1}) cannot be used for matrix completion. We next discuss how to design $\mu_t$ and $\alpha_t$ for matrix completion.

\subsection{Determining $\mu_{t}$}
The TARM algorithm is similar to SVP and NIHT as aforementioned. These three algorithms are all SVD based and a gradient descent step is involved at each iteration. The choice of descent step size $\mu_t$ is of key importance. In \cite{tanner2013normalized, wei2016guarantees}, $\mu_t$ are chosen adaptively based on the idea of the steepest descent. Due to the similarity between TARM and NIHT, we follow the methods in \cite{tanner2013normalized, wei2016guarantees} and choose $\mu_t$ as
\begin{align}\label{mu_mc}
	\mu_{t}&=\frac{\|\mathcal{P}_{\mathcal{S}}^{(t)}(\mathcal{A}^T(\y-\mathcal{A}(\X^{(t)})))\|_F^2}{\|\mathcal{A}(\mathcal{P}_{\mathcal{S}}^{(t)}(\mathcal{A}^T(\y-\mathcal{A}(\X^{(t)}))))\|_2^2}
\end{align}
where $\mathcal{P}^{(t)}_{\mathcal{S}}:\mathbb{R}^{n_1\times n_2}\rightarrow \mathcal{S}$ denotes a projection operator with $\mathcal{S}$ being a predetermined subspace of $\R^{n_1\times n_2}$. The subspace $\mathcal{S}$ can be chosen as the left singular vector space of $\X^{(t)}$, the right singular vector space of $\X^{(t)}$, or the tangent space of $C(\X)=\frac{1}{2}\|\y-\mathcal{A}(\X)\|_F^2$ at $\X=\X^{(t)}$. Let the SVD of $\X^{(t)}$ be $\X^{(t)}=\U^{(t)}\bm{\Sigma}^{(t)}(\V^{(t)})^T$. Then, the corresponding three projection operators are given respectively by
\BS\label{muformc}
\begin{align}
	\mathcal{P}^{(t)}_{\mathcal{S}_1}(\X)&=\U^{(t)}(\U^{(t)})^T\X\\
	\mathcal{P}^{(t)}_{\mathcal{S}_2}(\X)&=\X\V^{(t)}(\V^{(t)})^T\\
	\mathcal{P}^{(t)}_{\mathcal{S}_3}(\X)&=\U^{(t)}(\U^{(t)})^T\X+\X\V^{(t)}(\V^{(t)})^T\notag\\
	 &\ \ -\U^{(t)}(\U^{(t)})^T\X\V^{(t)}(\V^{(t)})^T.
\end{align}\ES
By combining (\ref{muformc}) with (\ref{mu_mc}), we obtain three different choices of $\mu_t$. Later, we present numerical results to compare the impact of different choices of $\mu_t$ on the performance of TARM.

\subsection{Determining $\alpha_{t}$ and $c_t$}
The linear combination parameters $\alpha_{t}$ and  $c_t$ in TARM is difficult to evaluate since Assumptions 1 and 2 do not hold for TARM in the matrix completion problem. Recall that $c_t$ is determined by $\alpha_t$ through (\ref{c_t}). So, we only need to determine $\alpha_t$. In the following, we propose three different approaches to evaluate $\alpha_t$.

The first approach is to choose $\alpha_t$ as in (\ref{alpha_t1}):
\begin{align}
	\alpha_t=\frac{\mr{div}(\mathcal{D}(\R^{(t)}))}{n}.
\end{align}
We use the Monte Carlo method to compute the divergence. Specifically, the divergence of $\mathcal{D}(\R^{(t)})$ can be estimated by \cite{metzler2016denoising}
\begin{align}\label{divergencec}
	\mr{div}(\mathcal{D}(\R^{(t)}))=\E_{\N}\left[\left<\frac{\mathcal{D}(\R^{(t)}+\epsilon \N)-\mathcal{D}(\R^{(t)})}{\epsilon},\N\right>\right]
\end{align}
where $\N\in \mathbb{R}^{n_1\times n_2}$ is a random Gaussian matrix with zero mean and unit variance entries, and $\epsilon$ is a small real number. The expectation in (\ref{divergencec}) can be approximated by sample mean. When the size of $\R^{(t)}$ is large, one sample is good enough for approximation. 
 
 We now describe the second approach. Recall that we choose $c_t$ according to (\ref{parameters}c) to satisfy Condition 2: $\left<\R^{(t)}-\X^{\ast},\X^{(t)}-\X^{\ast}\right>=0$. Since $\X^{\ast}$ is unknown, finding $\alpha_t$ to satisfy Condition 2 is difficult. Instead, we try to find $\alpha_t$ that minimizes the transformed correlation of the two estimation errors:
 \BS\label{esterr}
 \begin{align}
 	&\left|\left<\mathcal{A}(\R^{(t)}-\X^{\ast}),\mathcal{A}(\X^{(t)}-\X^{\ast})\right>\right|\\
 	=&\left|\left<\mathcal{A}(\R^{(t)})-\y,\mathcal{A}(\X^{(t)})-\y\right>\right|\\
 	=&\left|\left<\mathcal{A}(c_t(\Z^{(t)}-\alpha_t \R^{(t)}))-\y,\mathcal{A}(\R^{(t)})-\y\right>\right|\\
 	=&\left|\left<\frac{\left<\Z^{(t)}-\alpha_{t}\R^{(t)},\R^{(t)}\right>}{\|\Z^{(t)}-\alpha_{t}\R^{(t)}\|_F^2}\mathcal{A}(\Z^{(t)}-\alpha_t \R^{(t)})-\y,\mathcal{A}(\R^{(t)})-\y\right>\right|.
 \end{align}\ES
The minimization of (\ref{esterr}d) over $\alpha_t$ can be done by an exhaustive search over a small neighbourhood of zero.

The third approach is to set $\alpha_t$ as the asymptotic limit given in (\ref{ac}a). We next provide numerical simulations to show the impact of the above three different choices of $\alpha_t$ on the performance of TARM.

\subsection{Numerical Results}
We compare the TARM algorithms with different choices of $\mu_t$ and $\alpha_t$. We also compare TARM with the existing matrix completion algorithms, including SVP \cite{jain2010guaranteed}, NIHT \cite{tanner2013normalized}, and RGrad \cite{wei2016guarantees}. The matrix form $\A\in \mathbb{R}^{m\times n}$ of the matrix completion operator $\mathcal{A}$ is chosen as a random selection matrix (with randomly selected rows from a permutation matrix). The low-rank matrix $\X^{\ast}\in\mathbb{R}^{n_1\times n_2}$ is generated by the multiplication of two random Gaussian matrices of size $n_1\times r$ and $r \times n_2$.

\begin{figure}[!ht]
\centering
	\includegraphics[width=0.95\linewidth]{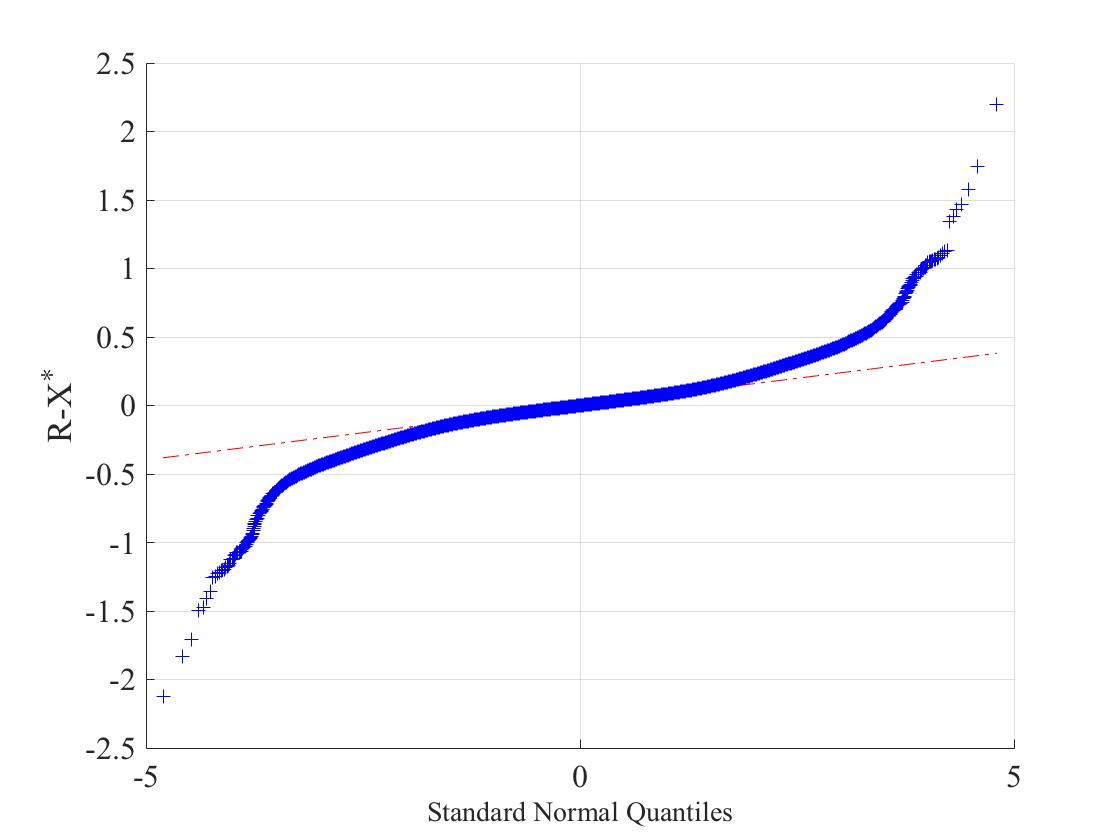}
	\caption{The QQplots of the output error of Module A in the 5th iteration of TARM for matrix completion. Simulation settings: $n_1=800,n_2=800, r=50, \frac{m}{n_1 n_2}=0.3, \sigma^2=0$.}\label{qqplot_mc}
\end{figure}

\begin{figure}[!ht]
\centering
	\includegraphics[width=0.95\linewidth]{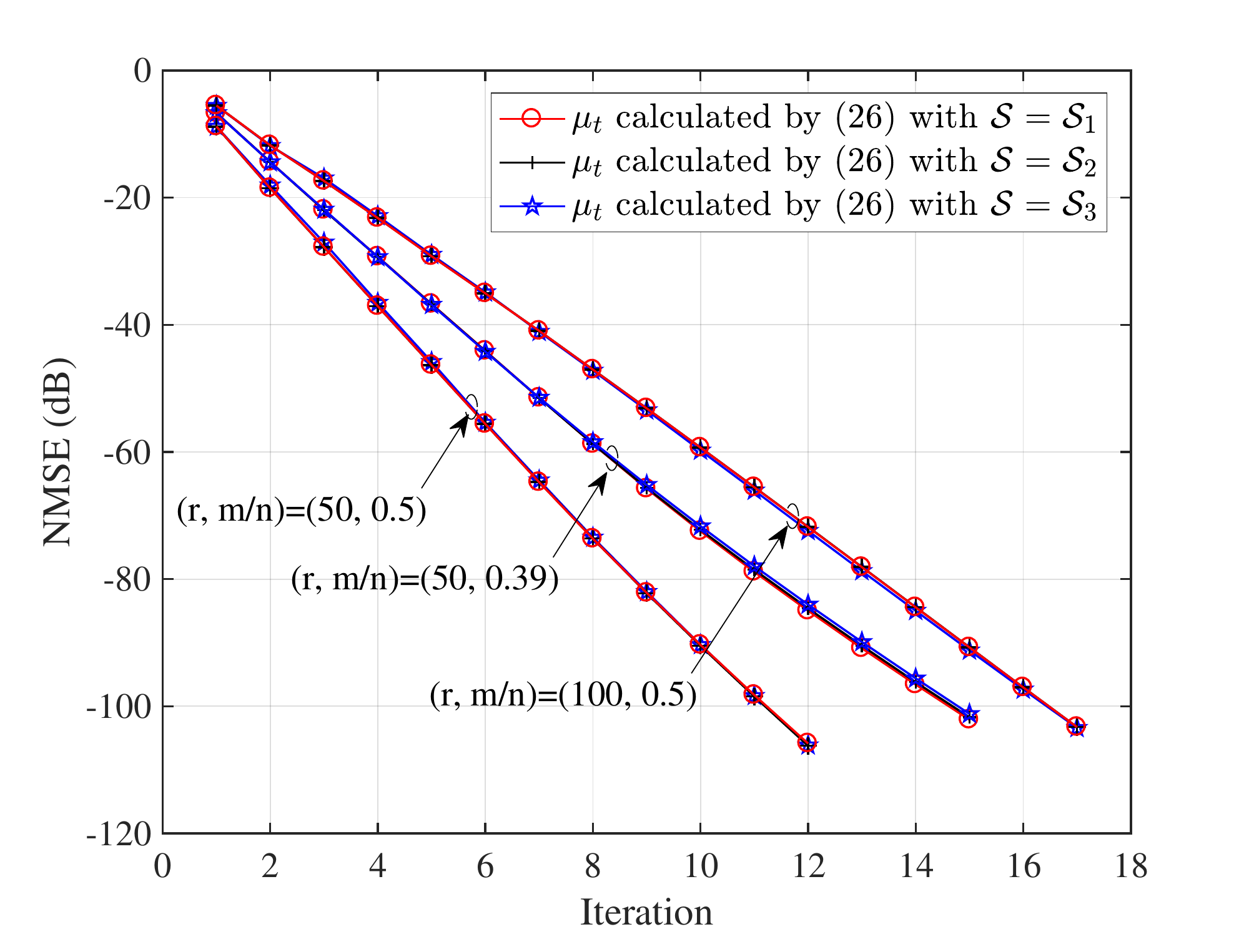}
	\caption{Comparison of the TARM algorithms for matrix completion with different choices of $\mu_t$. $n_1=n_2=1000,\sigma^2=0$.}\label{compare_mu}
\end{figure}

\begin{figure}[!ht]
\centering
	\includegraphics[width=0.95\linewidth]{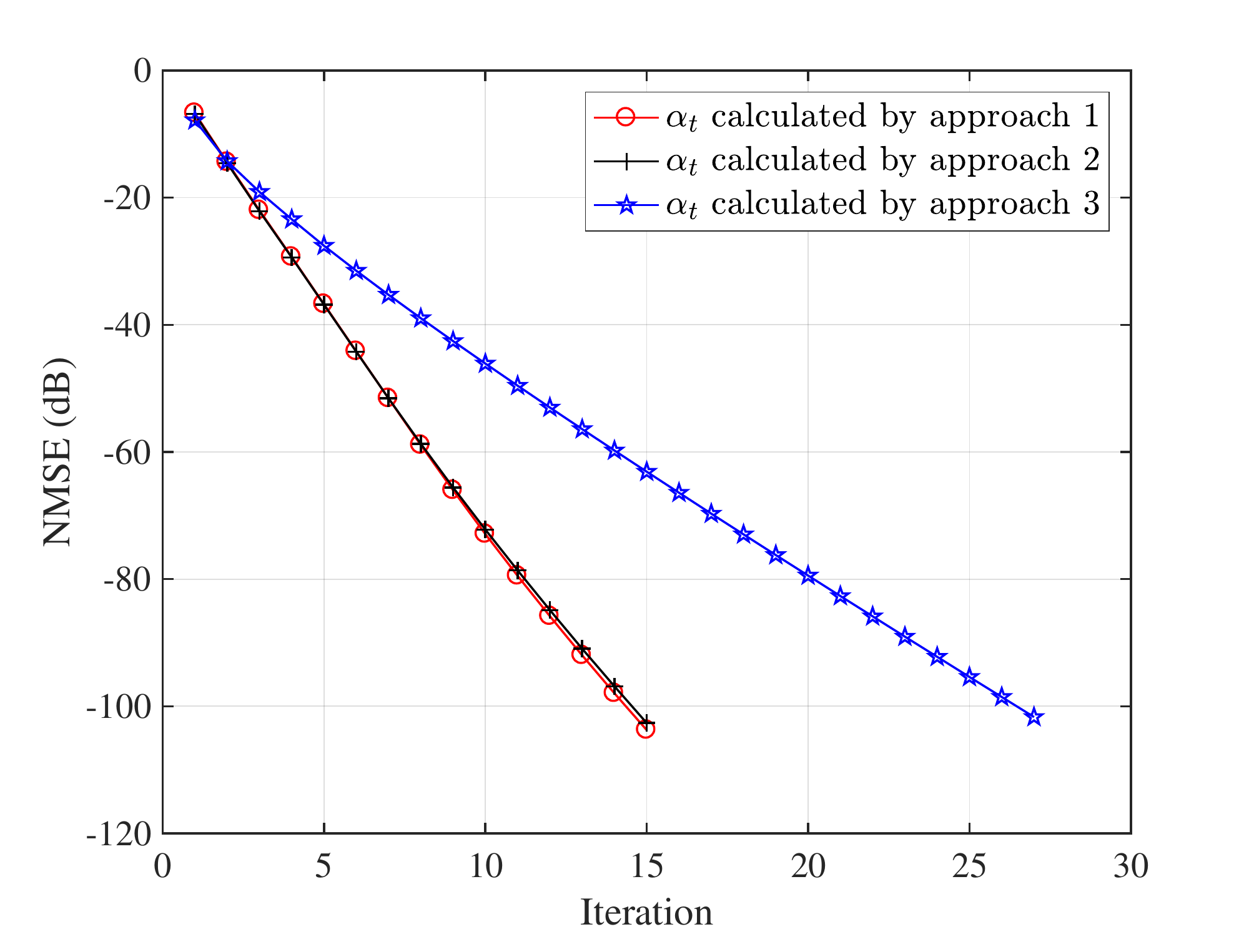}
	\caption{Comparison of the TARM algorithms for matrix completion with different choices of $\alpha_t$. $n_1=n_2=1000, r=50, \frac{m}{n_1 n_2}=0.39,\sigma^2=0$.}\label{compare_alpha}
\end{figure}
\subsubsection{Non-Gaussianity of the output error of Module A}In Fig. \ref{qqplot_mc}, we plot the QQplot of the input estimation errors of Module A of TARM for matrix completion. The QQplot shows that the distribution of the estimation errors of Module A is non-Gaussian. Thus, Assumption 2 does not hold for matrix completion.

\subsubsection{Comparisons of different choices of $\mu_t$}
We compare the TARM algorithms with $\mu_t$ in (\ref{mu_mc}) and the subspace $\mathcal{S}$ given by (\ref{muformc}), as shown in Fig. \ref{compare_mu}. We see that the performance of TARM is not sensitive to the three choices of $\mathcal{S}$ in (\ref{muformc}). In the following, we always choose $\mu_t$ with $\mathcal{S}$ given by (\ref{muformc}a).

\subsubsection{Comparisons of different choices of $\alpha_t$}
We compare the TARM algorithms with $\alpha_t$ given by the three different approaches in Subsection B. As shown in Fig. \ref{compare_alpha}, the first approach has the best performance among the three; the second approach performance close to the first one; the third approach performs considerably worse than the first two. Note that the second approach involves exhaustive search over $\alpha_t$, which is computationally involving. Thus, we henceforth choose $\alpha_t$ based on the first approach.

\subsubsection{Performance comparisons}
We compare TARM with the state-of-the-art algorithms for matrix completion. All the algorithms are run under the same settings for 100 random realizations. The numerical results are shown in Fig. \ref{plot_mc}. We see that TARM converges much faster than all the other algorithms.

\begin{figure}[!ht]
\centering
	\includegraphics[width=0.5\linewidth]{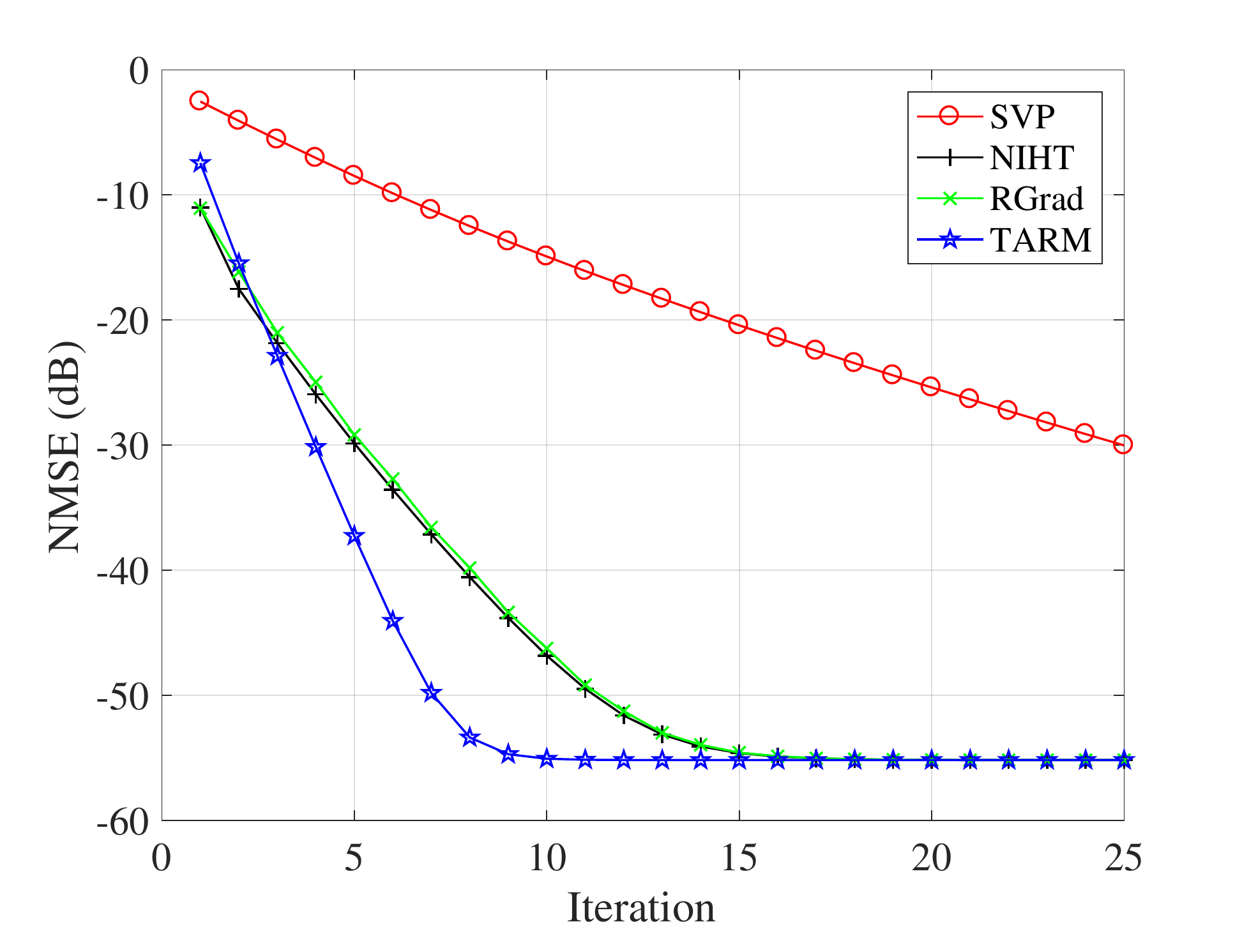}\includegraphics[width=0.5\linewidth]{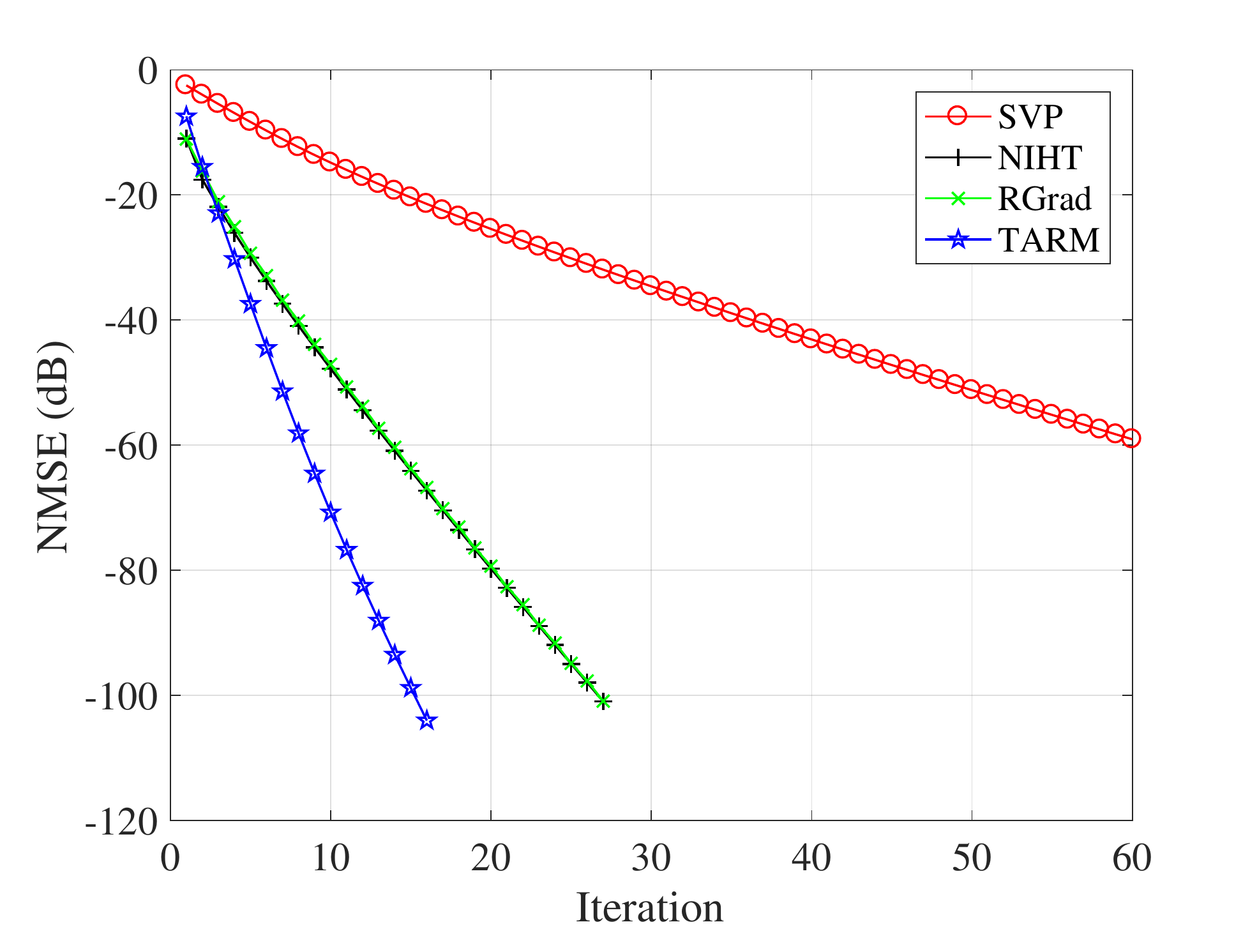}
	\caption{Comparison of algorithms for matrix completion. Left: $n_1=n_2=1000$, $r=50, m/n=0.39, \sigma^2=10^{-5}$. Right: $n_1=n_2=1000$, $r=50, m/n=0.39, \sigma^2=0$.}\label{plot_mc}
\end{figure}

\subsubsection{Empirical phase transition}
Similar to the case of low-rank matrix recovery. We consider an algorithm to be successful in recovering the low-rank matrix $\X^{\ast}$ when the following conditions are satisfied: 1) the normalized mean square error $\frac{\|\X^{(t)}-\X^{\ast}\|^2_F}{\|\X^{\ast}\|^2_F}\leq 10^{-6}$; 2) the iteration time $t<1000$. In Fig. \ref{pt_orth}, we plot the phase transition curves of the algorithms mentioned before. From Fig. \ref{plot_pr}, we see that the phase transition of TARM is the closest to the upper bound and considerably higher than the curves of NIHT and RGrad.

\begin{figure}[!h]
\centering
	\includegraphics[width=\linewidth]{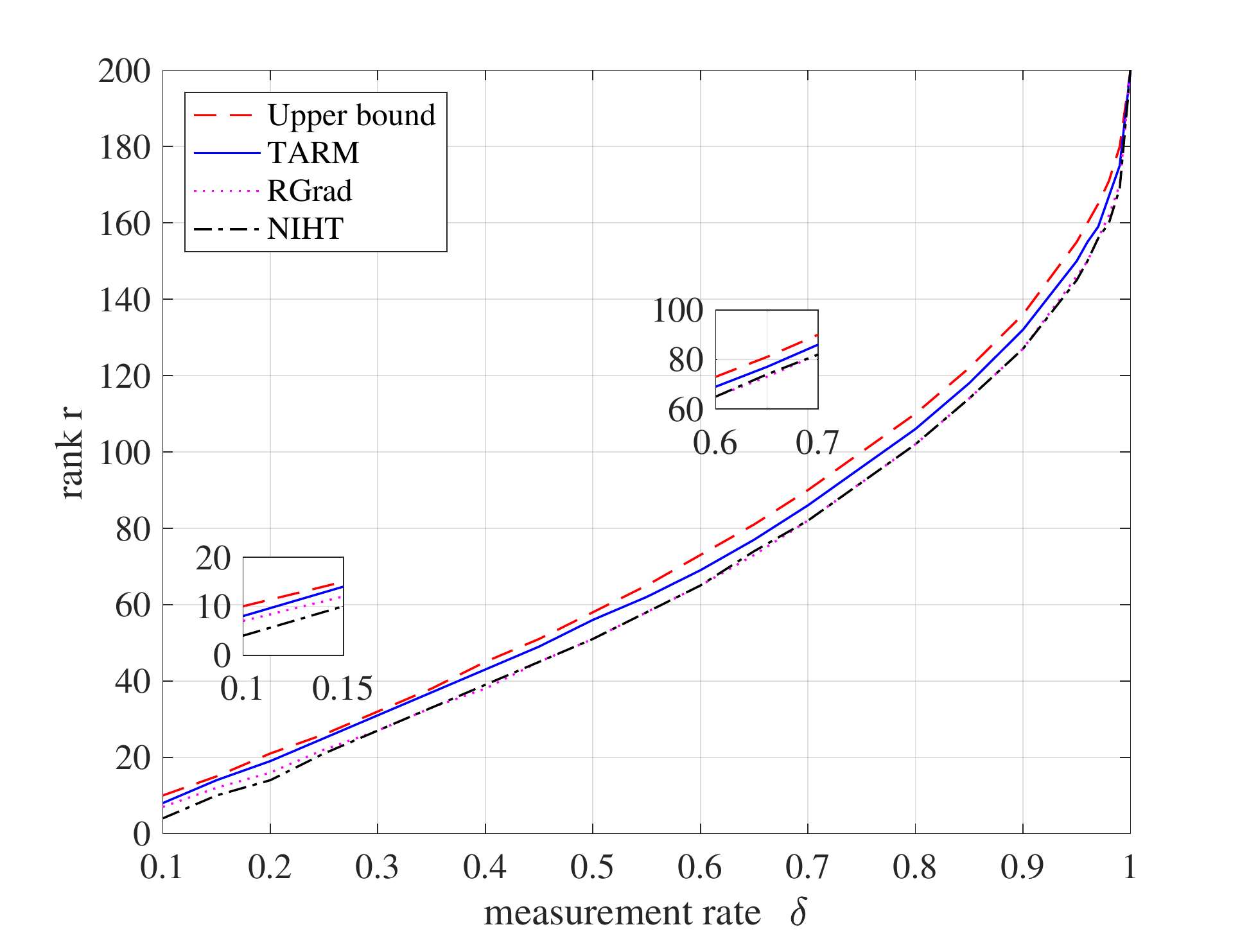}
	\caption{The phase transition curves of various matrix completion algorithms. $n_1=n_2=200, \sigma^2=0$. For each algorithm, the region below the phase transition curve corresponds to the successful recovery of $\X^{\ast}$.}\label{plot_pr}
\end{figure}

\section{Conclusions}
In this paper, we proposed a low-complexity iterative algorithm termed TARM for solving the stable ARM problem. The proposed algorithm can be applied to both low-rank matrix recovery and matrix completion. For low-rank matrix recovery, the performance of TARM can be accurately characterized by the state evolution technique when ROIL operators are involved. For matrix completion, we showed that, although state evolution is not accurate, the parameters of TARM can be carefully tuned to achieve good performance. Numerical results demonstrate that TARM has competitive performance compared to other existing algorithms for both low-rank matrix recovery and matrix completion.

\appendices
\section{Proof of Lemma 1}

We first determine $\mu_t$. We have
\BS
\begin{align}
&\left<\R^{(t)}-\X^{\ast},\X^{(t-1)}-\X^{\ast}\right>\notag\\
=&\left<\X^{(t\!-\!1)}\!\!+\!\!\mu_t\mathcal{A}^{T}(\y\!\!-\!\!\mathcal{A}(\X^{(t\!-\!1)}))\!\!-\!\!\X^{\ast},\X^{(t\!-\!1)}\!-\!\X^{\ast}\right>\\
=&\left<\X^{(t-1)}\!\!+\!\!\mu_t\mathcal{A}^{T}(\mathcal{A}(\X^{\ast})\!+\!\n-\!\!\mathcal{A}(\X^{(t-1)}))\!-\!\X^{\ast},\X^{(t-1)}\!-\!\X^{\ast}\right>\\
=&\left<\X^{(t-1)}-\X^{\ast},\X^{(t-1)}-\X^{\ast}\right>\\
&-\mu_{t}\left<\mathcal{A}^{T}(\mathcal{A}(\X^{(t-1)}-\X^{\ast})),\X^{(t-1)}-\X^{\ast}\right>\notag
	\\&+\mu_{t}\left<\mathcal{A}^{T}(\n),\X^{(t-1)}-\X^{\ast}\right>\notag\\
	=&\left<\X^{(t-1)}-\X^{\ast},\X^{(t-1)}-\X^{\ast}\right>\\
	&-\mu_{t}\left<\mathcal{A}(\X^{(t-1)}-\X^{\ast}),\mathcal{A}(\X^{(t-1)}-\X^{\ast})\right>\notag\\
	&+\mu_{t}\left<\n,\mathcal{A}(\X^{(t-1)}-\X^{\ast})\right>\notag
\end{align}\label{apda}\ES
where step (\ref{apda}a) follows by substituting $\R^{(t)}$ in Line 3 of Algorithm \ref{algorithm1}, and step (\ref{apda}d) follows by noting
\begin{align}
	\left<\mathcal{A}(\bm{B}),\bm{c}\right>=\left<\bm{B},\mathcal{A}^T(\bm{c})\right>
\end{align}
for any matrix $\bm{B}$ and vector $\bm{c}$ of appropriate sizes. Together with Condition 1, we obtain (\ref{parameters}a).

We next determine $\alpha_{t}$ and $c_{t}$. First note
\BS
\begin{align}
	\|\X^{(t)}-\R^{(t)}\|_F^2&=\|\X^{(t)}-\X^{\ast}\|_F^2+\|\X^{\ast}-\R^{(t)}\|_F^2\notag\\
	&\ \ \ +2\left<\X^{(t)}-\X^{\ast},\X^{\ast}-\R^{(t)}\right>\\
	&=\|\X^{(t)}\!-\!\X^{\ast}\|_F^2+\|\X^{\ast}\!-\!\R^{(t)}\|_F^2
\end{align}\label{appda2}\ES
where (\ref{appda2}b) is from Condition 2 in (\ref{cond2}). Recall that in the $t$-th iteration $\R^{(t)}$ is a function of $\mu_t$ but not of $\alpha_t$ and $c_t$. Thus, minimizing $\|\X^{(t)}-\X^{\ast}\|_F^2$ over $\alpha_t$ and $c_t$ is equivalent to minimizing $\|\X^{(t)}-\R^{(t)}\|_F^2$ over $\alpha_t$ and $c_t$. For any given $\alpha_t$, the optimal $c_t$ to minimize $\|\X^{(t)}-\R^{(t)}\|_F^2=\|c_t(\Z^{(t)}-\alpha_{t}\R^{(t)})-\R^{(t)}\|_F^2$ is given by
\begin{align}
	c_{t}=\frac{\left<\Z^{(t)}-\alpha_{t}\R^{(t)},\R^{(t)}\right>}{\|\Z^{(t)}-\alpha_{t}\R^{(t)}\|_F^2}.\label{cpara}
\end{align}
Then,
\BS
\begin{align}
	&\left<\X^{(t)}-\X^{\ast},\R^{(t)}-\X^{\ast}\right>\notag\\
	=&\left<c_t(\Z^{(t)}-\alpha_{t}\R^{(t)})-\X^{\ast},\R^{(t)}-\X^{\ast}\right>\\
	=&\left<\frac{\left<\Z^{(t)}-\alpha_{t}\R^{(t)},\R^{(t)}\right>}{\|\Z^{(t)}-\alpha_{t}\R^{(t)}\|_F^2}(\Z^{(t)}-\alpha_{t}\R^{(t)})-\X^{\ast},\R^{(t)}-\X^{\ast}\right>
\end{align}\label{alphaeqn}\ES
where (\ref{alphaeqn}a) follows by substituting $\X^{(t)}$ in Line 5 of Algorithm \ref{algorithm1}, and (\ref{alphaeqn}b) by substituting $c_t$ in (\ref{cpara}). Combining (\ref{alphaeqn}) and Condition 2, we see that
$\alpha_t$ is the solution of the following quadratic equation:
\begin{align}
	a_t\alpha_t^2+b_t\alpha_t+ d_t=0
\end{align}
where $a_t, b_t$, and $d_t$ are defined in (\ref{abd}).
Therefore, $\alpha_t$ is given by (\ref{parameters}a). With the above choice of $c_t$, we have
\begin{align}\label{orthogonal2}
	&\left<\X^{(t)}-\R^{(t)},\X^{(t)}\right>\notag\\
	=&\left<c_t(\Z^{(t)}\!-\!\alpha_{t}\R^{(t)})-\R^{(t)},c_t(\Z^{(t)}\!-\!\alpha_{t}\R^{(t)})\right>=0.
\end{align}
This orthogonality is useful in analyzing the performance of Module B.

\section{Convergence Analysis of TARM Based on RIP}
Without loss of generality, we assume $n_1\leq n_2$ in this appendix. Following the convention in \cite{tanner2013normalized}, we focus our discussion on the noiseless case, i.e., $\n=\bm{0}$.
\begin{definition} (Restricted Isometry Property).
	Given a linear operator $\mathcal{A}: \mathbb{R}^{n_1\times n_2}\rightarrow \mathbb{R}^m$, a minimum constant called the rank restricted isometry constant (RIC) $\delta_{r}(\mathcal{A})\in (0,1)$ exists such that 
\begin{align}
	(1\!-\!\delta_{r}(\mathcal{A}))\|\X\|_F^2\leq \|\gamma\mathcal{A}(\X)\|_2^2\leq (1\!+\!\delta_{r}(\mathcal{A}))\|\X\|_F^2
\end{align}
for all $\X\in \mathbb{R}^{n_1\times n_2}$ with $\mr{rank}(\X)\leq r$, where $\gamma>0$ is a constant scaling factor.
\end{definition}

We now introduce two useful lemmas.
\begin{lemma}\label{lemma3}
 	Assume that $\alpha_{t+1}$ and $c_{t+1}$ satisfy Condition 2 and Condition 3. Then,
 	\begin{align}
 		\|\X^{(t)}-\R^{(t)}\|_F^2=\frac{\|\R^{(t)}-\Z^{(t)}\|_F^2}{\frac{\|\R^{(t)}-\Z^{(t)}\|_F^2}{\|\Z^{(t)}\|_F^2}\alpha_t^2+(1-\alpha_t)^2}.
 	\end{align}
 \end{lemma}

 \begin{lemma}
 	Let $\Z^{(t)}$ be the best rank-$r$ approximation of $\R^{(t)}$. Then,
 	\begin{align}\label{lemma4}
 		\|\R^{(t)}-\Z^{(t)}\|_F^2\leq\|\X^{\ast}-\R^{(t)}\|_F^2.
 	\end{align}
 \end{lemma}
 The proof of Lemma 3 is given in Appendix C. Lemma 4 is straightforward from the definition of the best rank-$r$ approximation \cite[p. 211-218]{eckart1936approximation}.

\begin{theorem}
 	Assume that $\mu_t, \alpha_t, c_t$ satisfy Conditions 1-3, and the linear operator $\mathcal{A}$ satisfies the RIP with rank $n_1$ and RIC $\delta_{n_1}$. Then,
 \begin{align}
 		\|\X^{(t)}\!-\!\X^{\ast}\|_F^2\!\leq\!\left(\!\frac{1}{(1\!-\!\alpha_{t})^2}\!-\!1\!\right)\!\!\left(\!\frac{1\!+\!\delta_{n_1}}{1\!-\!\delta_{n_1}}\!-\!1\!\right)^2\!\! \|\X^{(t-1)}\!-\!\X^{\ast}\|_F^2\label{theorem_iter}
 	\end{align}
TARM guarantees to converge when RIC satisfies $\alpha_t\neq 1, \forall t$, and
  \begin{align}
	\delta_{n_1}<\frac{1}{1+2\sqrt{\frac{1}{\xi}\left(\frac{1}{(1-\alpha_{max})^2}-1\right)}}
\end{align}
where the constant $\xi$ satisfies $0<\xi <1$, and $\alpha_{max}=\sup
   \{\alpha_t\}$.
\end{theorem}

\begin{proof}
	Since $\Z^{(t)}$ is the best rank-$r$ approximation of $\R^{(t)}$, we have $\|\R^{(t)}\|_F^2\geq\|\Z^{(t)}\|_F^2$. Then, from Lemma 3, we obtain
	\begin{align}\label{xrleq}
		\|\X^{(t)}-\R^{(t)}\|_F^2\leq\frac{\|\R^{(t)}-\Z^{(t)}\|_F^2}{(1-\alpha_t)^2}.
	\end{align}
Then, we have
\BS\label{xrerror}
\begin{align}
	\|\X^{(t)}-\R^{(t)}\|_F^2&=\|\X^{(t)}-\X^{\ast}+\X^{\ast}-\R^{(t)}\|_F^2\\
	&=\|\X^{(t)}-\X^{\ast}\|_F^2+\|\X^{\ast}-\R^{(t)}\|_F^2\notag\\
	&\ \ \ +2\left<\X^{(t)}-\X^{\ast},\X^{\ast}-\R^{(t)}\right>\\
	&=\|\X^{(t)}-\X^{\ast}\|_F^2+\|\X^{\ast}-\R^{(t)}\|_F^2
\end{align}\ES
where (\ref{xrerror}c) follows from $\left<\X^{(t)}-\X^{\ast},\X^{\ast}-\R^{(t)}\right>=0$ in Condition 2, and (\ref{xrerror}d) follows from (\ref{xrleq}). Combining (\ref{lemma4}), (\ref{xrleq}), and (\ref{xrerror}), we obtain
\BS\label{ineq0}
\begin{align}
&\|\X^{(t)}-\X^{\ast}\|_F^2\leq\left(\!\frac{1}{(1-\alpha_{t})^2}\!-\!1\!\right)\|\R^{(t)}-\X^{\ast}\|_F^2\\
=&\left(\!\frac{1}{(1-\alpha_{t})^2}\!-\!1\!\right)\|\X^{(t\!-\!1)}\!+\!\mu_t\mathcal{A}^{\ast}(\y-\mathcal{A}(\X^{(t-1)}))-\X^{\ast}\|_F^2\\
	=&\left(\!\frac{1}{(1-\alpha_{t})^2}\!-\!1\!\right)\|(\mathcal{I}\!-\!\mu_t\mathcal{A}^{\ast}\mathcal{A})(\X^{(t-1)}\!-\!\X^{\ast})\|_F^2.
\end{align}\ES
Since $\mathcal{A}$ has RIP with rank $n_1$ and RIC $\delta_{n_1}$, we obtain the following inequality from \cite{kyrillidis2014matrix}:
\begin{align}
	&\|(\mathcal{I}\!-\!\mu_t \mathcal{A}^{\ast}\!\mathcal{A})(\X^{(t\!-\!1)}\!\!-\!\!\X^{\ast})\|_F^2\notag\\
	\leq&  \max\left((\mu_t(1\!+\!\delta_{n_1})\!-\!1)^2,(\mu_t(1\!-\!\delta_{n_1})\!-\!1)^2\right) \|\X^{(t\!-\!1)}\!\!-\!\!\X^{\ast}\|_F^2.\label{ineq2}
\end{align}

Recall that $\mu_{t}=\frac{\|\X^{(t-1)}-\X^{\ast}\|_F^2}{\|\mathcal{A}(\X^{(t-1)}-\X^{\ast})\|_2^2}$ obtained by letting $\n=\bm{0}$ in (\ref{parameters}a). From RIP, we have
\begin{align}\label{48}
	\frac{1}{1+\delta_{n_1}}\leq \mu_t=\frac{\|\X^{(t-1)}-\X^{\ast}\|_F^2}{\|\mathcal{A}(\X^{(t-1)}-\X^{\ast})\|_2^2} \leq \frac{1}{1-\delta_{n_1}}.
\end{align}
Then, combining (\ref{ineq2}) and (\ref{48}), we have
\begin{align}
	\|(\mathcal{I}\!-\!\mu_t\mathcal{A}^{\ast}\mathcal{A})(\X^{(t\!-\!1)}\!-\!\X^{\ast})\|_F^2\leq \left(\frac{1\!+\!\delta_{n_1}}{1\!-\!\delta_{n_1}}\!-\!1\right)^2 \!\!\|\X^{(t\!-\!1)}\!-\!\X^{\ast}\|_F^2.\label{ineq4}
\end{align}
Combining (\ref{ineq4}) and (\ref{ineq0}), we arrive at (\ref{theorem_iter}). 

When $\delta_{n_1}$ satisfies (41), we have
\begin{align}
	\|\X^{(t)}-\X^{\ast}\|_F^2<\xi \|\X^{(t-1)}-\X^{\ast}\|_F^2
\end{align}
at each iteration $t$. Then, TARM converges exponentially to $\X^{\ast}$.
\end{proof}
We now compare the convergence rate of TARM with those of SVP and NIHT. Compared with \cite[Equ. 2.11-2.14]{tanner2013normalized}, (\ref{theorem_iter}) contains an extra term $\frac{1}{(1-\alpha_{t})^2}-1$. From numerical experiments, $\alpha_t$ is usually close to zero, implying that TARM converges much faster than SVP and NIHT.

\section{Proof of Theorem 1}
For a partial orthogonal ROIL operator $\mathcal{A}$, the following properties hold:
\BS\label{proporth}
\begin{align}
	\mathcal{A}(\mathcal{A}^T(\bm{a}))&=\bm{a}\\
	\left<\mathcal{A}^T(\bm{a}),\mathcal{A}^T(\bm{b})\right>&=\left<\bm{a},\bm{b}\right>.
\end{align}\ES
Then as $m,n\rightarrow \infty$ with $\frac{m}{n}\rightarrow \delta$, we have
\BS\label{error1}
\begin{align}
	&\left\|\R^{(t)}-\X^{\ast}\right\|_F^2\notag\\
	=&\left\|\X^{(t)}\!-\!\X^{\ast}\!-\!\frac{1}{\delta}\mathcal{A}^T\mathcal{A}(\X^{(t)}\!-\!\X^{\ast})+\mu_t\mathcal{A}^T(\n)\right\|_F^2\\
	=&\|\X^{(t)}\!-\!\!\X^{\ast}\|_F^2\!+\!\frac{1}{\delta^2}\|\mathcal{A}^T\!\mathcal{A}(\X^{(t)}\!-\!\X^{\ast})\|_F^2\notag\\
	&-\frac{2}{\delta}\|\mathcal{A}(\X^{(t)}\!-\!\X^{\ast})\|_F^2+\frac{1}{\delta^2}\|\n\|_2^2\\
	=&\|\X^{(t)}\!-\!\X^{\ast}\|_F^2\!+\!\frac{1}{\delta^2}\|\mathcal{A}(\X^{(t)}\!-\!\X^{\ast})\|_F^2\notag\\
	&-\frac{2}{\delta}\|\mathcal{A}(\X^{(t)}-\X^{\ast})\|_F^2+\frac{1}{\delta^2}\|\n\|_2^2 \\
	=&\|\X^{(t)}-\X^{\ast}\|_F^2+\frac{1}{\delta}\|\X^{(t)}-\X^{\ast}\|_F^2\notag\\
	&-2\|\X^{(t)}-\X^{\ast}\|_F^2+\frac{1}{\delta^2}\|\n\|_2^2\\
	=&\left(\frac{1}{\delta}-1\right)\|\X^{(t)}-\X^{\ast}\|_F^2+n\sigma^2
\end{align}\ES
where (\ref{error1}a) is obtained by substituting $\R^{(t)}=\X^{(t-1)}+\mu_t\mathcal{A}^{T}(\y-\mathcal{A}(\X^{(t-1)}))$ and $\y=\mathcal{A}(\X^{\ast})+\n$, (\ref{error1}b) is obtained by noting that $\n$ is independent of $\mathcal{A}(\X^{(t)}-\X^{\ast})$ (ensured by Assumption 1), (\ref{error1}c) follows from (\ref{proporth}b), and (\ref{error1}d) follows from $\frac{\|\mathcal{A}(\X^{(t)}-\X^{\ast})\|_2^2}{\|\X^{(t)}-\X^{\ast}\|_F^2}\rightarrow\delta$ (see (\ref{mu})). When $\frac{1}{n}\|\X^{(t)}-\X^{\ast}\|_F^2\rightarrow\tau$, we have $\frac{1}{n}\|\R^{(t)}-\X^{\ast}\|_F^2\rightarrow (\frac{1}{\delta}-1)\tau +\sigma^2$.

We now consider the case of Gaussian ROIL operators. As $m,n \rightarrow \infty$ with $\frac{m}{n	}\rightarrow\delta$, we have

\BS\label{error2}
\begin{align}
	&\left\|\R^{(t)}-\X^{\ast}\right\|_F^2\notag\\
	=&\|\X^{(t)}\!-\!\!\X^{\ast}\|_F^2+\frac{1}{\delta^2}\|\mathcal{A}^T\!\mathcal{A}(\X^{(t)}\!-\!\X^{\ast})\|_F^2\notag\\
	&-\!\frac{2}{\delta}\|\mathcal{A}(\X^{(t)}\!\!-\!\!\X^{\ast})\|_F^2\!+\!\frac{1}{\delta^2}\|\n\|_2^2\\
	=&\|\!\X^{(t)}\!-\!\!\X^{\ast}\|_F^2\!+\!\!\frac{1}{\delta^2}\|\!\A^T \!\A\mr{vec}(\X^{(t)}\!\!-\!\!\X^{\ast})\|_F^2\notag\\
	&-\frac{2}{\delta}\|\mathcal{A}(\X^{(t)}\!\!-\!\!\X^{\ast})\|_F^2\!\!+\!\!\frac{1}{\delta^2}\|\n\|_2^2\\
	=&\|\!\X^{(t)}\!\!-\!\!\X^{\ast}\!\|_F^2\!\!+\!\!\frac{1}{\delta^2}\frac{\|\!\A^T \!\A\|_F^2}{mn}\|\mr{vec}(\!\X^{(t)}\!\!-\!\!\X^{\ast}\!)\|_2^2\notag\\
	&-\frac{2}{\delta}\|\!\mathcal{A}(\!\X^{(t)}\!\!-\!\!\X^{\ast}\!)\!\|_F^2\!\!+\!\!\frac{1}{\delta^2}\|\n\|_2^2\\
	=&\|\!\X^{(t)}\!\!-\!\!\X^{\ast}\!\|_F^2\!\!+\!\!\frac{1}{\delta^2}\frac{\mr{Tr}{(\!(\A^T\!\A)\!^2\!)}}{mn}\|\!\X^{(t)}\!\!-\!\!\X^{\ast}\!\|_F^2\notag\\
	&-\frac{2}{\delta}\|\mathcal{A}(\!\X^{(t)}\!\!-\!\!\X^{\ast}\!)\|_F^2\!\!+\!\!\frac{1}{\delta^2}\|\n\|_2^2\\
	=&\|\X^{(t)}\!-\!\X^{\ast}\|_F^2\!+\!(1+\frac{1}{\delta})\|\X^{(t)}\!-\!\X^{\ast}\|_F^2\notag\\
	&-\frac{2}{\delta}\|\mathcal{A}(\X^{(t)}\!-\!\X^{\ast})\|_F^2\!+\!\frac{1}{\delta^2}\|\n\|_2^2\\
	=&\|\X^{(t)}-\X^{\ast}\|_F^2+(1+\frac{1}{\delta})\|\X^{(t)}-\X^{\ast}\|_F^2\notag\\
	&-2\|\X^{(t)}-\X^{\ast}\|_F^2+\frac{1}{\delta^2}\|\n\|_2^2\\
	=&\frac{1}{\delta}\|\X^{(t)}-\X^{\ast}\|_F^2+n\sigma^2
\end{align}\ES
where (\ref{error2}a) is from (\ref{error1}c), (\ref{error2}b) follows by utilizing the matrix form of $\mathcal{A}$, (\ref{error2}c) follows from the fact that $\V_A$ is a Haar distributed orthogonal matrix independent of $\X^{(t)}-\X^{\ast}$, (\ref{error2}e) is obtained by noting that $\frac{1}{mn}\mr{Tr}{((\A^T\A)^2)}\rightarrow \delta+\delta^2$ since $\A^T\A$ is a Wishart matrix with variance $\frac{1}{n}$ \cite[p.26]{tulino2004random}, and (\ref{error2}f) follows by noting $\frac{\|\mathcal{A}(\X^{(t)}-\X^{\ast})\|_2^2}{\|\X^{(t)}-\X^{\ast}\|_F^2}\rightarrow\delta$. When $\frac{1}{n}\|\X^{(t)}-\X^{\ast}\|_F^2\rightarrow\tau$, we have $\frac{1}{n}\|\R^{(t)}-\X^{\ast}\|_F^2\rightarrow \frac{1}{\delta}\tau +\sigma^2$.

\section{Proof of Lemma 2}
We first introduce two useful facts.

Fact 1: When $n_1,n_2\rightarrow \infty$ with $n_1/n_2=\rho, r/n_2=\lambda$, and the singular value of $\frac{1}{\sqrt{n_2}}\X^{\ast}$ are $[\theta_1, \theta_2,\cdots,\theta_{r}]$, the $i$-th singular value $\sigma_i$ of the Gaussian noise corrupted matrix $\R^{(t)}$ is given by \cite[equ. 9]{benaych2012singular}
\begin{align}\label{noisesingular}
	\begin{split}
		\sigma_i\xrightarrow[]{\text{a.s.}}\begin{cases}
		\sqrt{n_2\frac{(v_t+\theta_i^2)(\rho v_t+\theta_i^2)}{\theta_i^2}}& \text{if } i\leq r \text{ and } \theta_i>\rho^{1/4}\\
		\sqrt{n_2 v_t}(1+\sqrt{\rho})& \text{otherwise}
\end{cases}
	\end{split}
\end{align}
where $v_t$ is the variance of the Gaussian noise.

 
Fact 2: From \cite[equ. 9]{candes2013unbiased}, the divergence of a spectral function $h(\R)$ is given by
\begin{align}\label{spectraldiv}
		\text{div}(h(\R))&=|n_1-\!n_2|\!\sum_{i=1}^{\min(n_1,n_2)}\frac{h_i(\sigma_i)}{\sigma_i}+\sum_{i=1}^{\min(n_1,n_2)}h_i'(\sigma_i)\notag\\&\ \ \ +2\sum_{i\neq j,i,j=1}^{\min(n_1,n_2)}\frac{\sigma_i h_i(\sigma_i)}{\sigma_i^2-\sigma_j^2}.
\end{align}
The best rank-$r$ approximation denoiser $\mathcal{D}(\R)$ is a spectral function with
\begin{align}\label{rankrapprox}
	\begin{cases}
		h_i(\sigma_i)=\sigma_i & i\leq r;\\
		h_i(\sigma_i)=0 & i>r.
	\end{cases}
\end{align}
Combining (\ref{spectraldiv}) and (\ref{rankrapprox}), the divergence of $\mathcal{D}(\R^{(t)})$ is given by
\begin{align}\label{div1}
		\text{div}(\mathcal{D}(\R^{(t)}))&\!=\!|n_1\!-\!n_2|r\!+\!r^2\!+\!2\sum_{i=1}^{r}\!\!\sum_{j=r+1}^{\min{(n_1,n_2)}}\!\!\frac{\sigma_i^2}{\sigma_i^2\!-\!\sigma_j^2}.
\end{align}
Further, we have
\BS\label{div2}
\begin{align}
	&\ \ \ \ \ \sum_{i=1}^{r}\sum_{j=r+1}^{\min{(n_1,n_2)}}\frac{\sigma_i^2}{\sigma_i^2-\sigma_j^2}\notag\\
	&\stackrel{a.s.}{\rightarrow}(\min(n_1,n_2)-r)\sum_{i=1}^{r}\frac{\sigma_i^2}{\sigma_i^2-(\sqrt{n_2 v_t}(1+\sqrt{\rho}))^2}\\
	&=(\min(n_1,n_2)-r)\sum_{i=1}^{r}\frac{n_2\frac{(v_t+\theta_i^2)(\rho v_t+\theta_i^2)}{\theta_i^2}}{\frac{n_2(v_t+\theta_i^2)(\rho v_t+\theta_i^2)}{\theta_i^2}-n_2 v_t(1+\sqrt{\rho})^2}\\
	&=(\min(n_1,n_2)-r)\sum_{i=1}^{r}\frac{(v_t+\theta_i^2)(\rho v_t+\theta_i^2)}{(\sqrt{\rho}v_t-\theta_i^2)^2}\\
	&\stackrel{a.s.}{\rightarrow}(\min(n_1,n_2)-r)r\int_0^{\infty}\frac{(v_t+\theta^2)(\rho v_t+\theta^2)}{(\sqrt{\rho}v_t-\theta^2)^2} p(\theta)d\theta\\
	&=(\min(n_1,n_2)-r)r\Delta_1(v_t)
\end{align}\ES
where both (\ref{div2}a) and (\ref{div2}b) are from (\ref{noisesingular}), and (\ref{div2}e) follows by the definition of $\Delta_1(v_t)$. Combining (\ref{div1}) and (\ref{div2}), we obtain the asymptotic divergence of $\mathcal{D}(\R)$ given by
\BS
\begin{align}
	\text{div}(\mathcal{D}(\R))\stackrel{a.s.}{\rightarrow}&|n_1-\!n_2|r+r^2+2(\min(n_1,n_2)-r)r\Delta_1(v_t)
\end{align}\ES
and $\alpha_t=\frac{1}{n}\text{div}(f(\R^{(t)}))\stackrel{a.s.}{\rightarrow}\left|1-\frac{1}{\rho}\right|\lambda+\frac{1}{\rho}\lambda^2+2\left(\min\left(1,\frac{1}{\rho}\right)-\frac{\lambda}{\rho}\right)\lambda\Delta_1(v_t)=\alpha(v_t)$.

Recall that $\Z^{(t)}$ is the best rank-$r$ approximation of $\R^{(t)}$ satisfying
\BS\label{distance1}
\begin{align}
		\|\Z^{(t)}\|_F^2&=\sum_{i=1}^r \sigma_i^2\\
		\|\R^{(t)}\|_F^2-\|\Z^{(t)}\|_F^2&=\sum_{i=r+1}^{n_1} \sigma_i^2.
\end{align}\ES
Then, when $m,n\rightarrow\infty$ with $\frac{m}{n}\rightarrow \delta$, we have
\BS\label{distance3}
\begin{align}
	\|\Z^{(t)}\|_F^2&=\sum_{i=1}^r \sigma_i^2\\
	&\xrightarrow[]{\text{a.s.}}n_2\sum_{i=1}^r\frac{(v+\theta_i^2)(\rho v+\theta_i^2)}{\theta_i^2} \\
	&=n+\lambda\left(1+\frac{1}{\rho}\right)nv+\lambda n v^2\frac{1}{r}\sum_{i=1}^r \frac{1}{\theta_i^2}
\end{align}\ES
and
\BS\label{distance2}
\begin{align}
		&\ \ \ \  \|\R^{(t)}\|_F^2-\|\Z^{(t)}\|_F^2\notag\\
		&=\|\X^{\ast}\|_F^2+n v_t-\|\Z^{(t)}\|_F^2\\
		&\xrightarrow[]{\text{a.s.}}nv_t-\lambda(1+\frac{1}{\rho})n v_t-\lambda n v_t^2\frac{1}{r}\sum_{i=1}^r \frac{1}{\theta_i^2}
\end{align}\ES
where (\ref{distance3}b) is from (\ref{noisesingular}), (\ref{distance2}a) is from Assumption 2, and (\ref{distance2}b) is from (\ref{distance3}).
Then,
\BS\label{c_proof}
\begin{align}
	c_t=&\frac{\left<\Z^{(t)}-\alpha_{t}\R^{(t)},\R^{(t)}\right>}{\|\Z^{(t)}-\alpha_{t}\R^{(t)}\|_F^2}\\
	=&\frac{\left<\Z^{(t)},\R^{(t)}\right>-\alpha_{t}\|\R^{(t)}\|_F^2}{\|\Z^{(t)}\|_F^2-2\alpha_{t}\left<\Z^{(t)},\R^{(t)}\right>+\alpha_{t}^2\|\R^{(t)}\|_F^2}\\
	\xrightarrow[]{\text{a.s.}}&\frac{\|\Z^{(t)}\|_F^2-\alpha_t (n+v_t n)}{\|\Z^{(t)}\|_F^2-2\alpha_t\|\Z^{(t)}\|_F^2+\alpha_t^2(n+v_t n)}\\
	=&\frac{n+\lambda(1+\frac{1}{\rho})nv_t+\lambda n v_t^2\frac{1}{r}\sum_{i=1}^r \frac{1}{\theta_i^2}-\alpha_t (n+v_t n)}{(1\!-\!2\alpha_t)(n\!+\!\lambda(1\!+\!\frac{1}{\rho})nv_t\!+\!\lambda n v_t^2\frac{1}{r}\sum_{i=1}^r \frac{1}{\theta_i^2})\!+\!\alpha_t^2(n\!+\!v_t n)}\\
	=&\frac{1+\lambda(1+\frac{1}{\rho})v_t+\lambda  v_t^2\frac{1}{r}\sum_{i=1}^r \frac{1}{\theta_i^2}-\alpha_t (1+v_t)}{(1\!-\!2\alpha_t)(1\!+\!\lambda(1+\frac{1}{\rho})v_t\!+\!\lambda v_t^2\frac{1}{r}\sum_{i=1}^r \frac{1}{\theta_i^2})\!+\!\alpha_t^2(1\!+\!v_t)}\\
	\xrightarrow[]{\text{a.s.}}&\frac{1+\lambda(1+\frac{1}{\rho})v_t+\lambda  v_t^2\Delta_2-\alpha(v_t)(1+v_t)}{(1\!-\!2\alpha(v_t))(1\!+\!\lambda(1\!+\!\frac{1}{\rho})v_t\!+\!\lambda v_t^2\Delta_2)\!+\!(\alpha(v_t))^2(1+v_t)}\\
	=&c(v_t)
\end{align}\ES
where (\ref{c_proof}a) is from (\ref{parameters}c), (\ref{c_proof}c) follows from Assumption 2 that $\R^{(t)}=\X^{\ast}+\sqrt{v_t}\W$ with $\|\X^{\ast}\|_F^2=n$, and the elements of $\W$ independently drawn from $\mathcal{N}(0,1)$, (\ref{c_proof}d) is from (\ref{distance1}), and (\ref{c_proof}f) is from the definition of $\Delta_2$.

\section{Proof of Lemma 3}
 Recall the following orthogonality relations:
 \BS\label{orthogonalproper}
 \begin{align}
 	\left<\R^{(t)}-\Z^{(t)},\Z^{(t)}\right>&=0\\
 	\left<\R^{(t)}-\X^{(t)},\X^{(t)}\right>&=0
 \end{align}\ES
 where (\ref{orthogonalproper}a) follows from $\Z^{(t)}=\mathcal{D}(\R^{(t)})$ and $\mathcal{D}(\cdot)$ is the best rank-$r$ approximation denoiser, and (\ref{orthogonalproper}b) follows from (\ref{orthogonal2}).
 
With the above properties, we have
 \BS\label{lemma3proof}
\begin{align}
	&\ \ \ \ \|\X^{(t)}-\R^{(t)}\|_F^2\notag\\
	&=\|\R^{(t)}\|_F^2-\|\X^{(t)}\|_F^2\\
	&=\|\R^{(t)}\|_F^2-\left\|c_t(\Z^{(t)}-\alpha_{t}\R^{(t)})\right\|_F^2\\
	&=\|\R^{(t)}\|_F^2-\frac{\left<\Z^{(t)}-\alpha_{t}\R^{(t)},\R^{(t)}\right>^2}{\|\Z^{(t)}-\alpha_{t}\R^{(t)}\|_F^2}\\
	&=\frac{\|\R^{(t)}\|_F^2\|\Z^{(t)}\!-\!\alpha_{t}\R^{(t)}\|_F^2\!-\!\left<\!\Z^{(t)}\!-\!\alpha_{t}\R^{(t)},\R^{(t)}\!\right>^2}{\|\Z^{(t)}-\alpha_{t}\R^{(t)}\|_F^2}\\
	&=\frac{\|\R^{(t)}\|_F^2\|\Z^{(t)}\|_F^2-\left<\Z^{(t)},\R^{(t)}\right>^2}{\|\Z^{(t)}-\alpha_{t}\R^{(t)}\|_F^2}\\
	&=\frac{\|\R^{(t)}\|_F^2\|\Z^{(t)}\|_F^2-\|\Z^{(t)}\|_F^4}{\|\Z^{(t)}-\alpha_{t}\R^{(t)}\|_F^2}\\
	&=\frac{\|\R^{(t)}\|_F^2-\|\Z^{(t)}\|_F^2}{\frac{\|\R\|_F^2-\|\Z^{(t)}\|_F^2}{\|\Z^{(t)}\|_F^2}\alpha_t^2+(1-\alpha_t)^2}\\
	&=\frac{\|\R^{(t)}-\Z^{(t)}\|_F^2}{\frac{\|\R\|_F^2-\|\Z^{(t)}\|_F^2}{\|\Z^{(t)}\|_F^2}\alpha_t^2+(1-\alpha_t)^2}
\end{align}\ES
where (\ref{lemma3proof}a) follows from (\ref{orthogonalproper}b), (\ref{lemma3proof}b) follows by substituting $\X^{(t)}$ in Line 5 of Algorithm \ref{algorithm1}, (\ref{lemma3proof}c) follows by substituting $c_t$ in (\ref{parameters}c), and (\ref{lemma3proof}f-\ref{lemma3proof}h) follow from (\ref{orthogonalproper}a). This concludes the proof of Lemma \ref{lemma3}.

\section{Proof of Theorem 2}
From Condition 2 in (\ref{cond2}) and Assumption 2, we have\footnote{In fact, as $n_1,n_2,r\rightarrow\infty$ with $\frac{n_1}{n_2}\rightarrow\rho$ and $\frac{r}{n_2}\rightarrow \lambda$, the approximation in (\ref{alpha_t}) become accurate, i.e. $\alpha_t=\frac{1}{n}\text{div}(\mathcal{D}(\R^{(t)}))$ asymptotically satisfies Condition 2. Thus, (\ref{orthoproperty2}b) asymptotically holds.}
\BS\label{orthoproperty2}
\begin{align}
	\left<\R^{(t)}-\X^{\ast},\X^{\ast}\right>&=0\\
	\left<\R^{(t)}-\X^{\ast},\X^{(t)}-\X^{\ast}\right>&=0.
\end{align}\ES

Then,
\BS\label{errorxxstar}
\begin{align}
	&\|\X^{(t)}-\X^{\ast}\|_F^2\notag\\
	=&\|\X^{(t)}-\R^{(t)}\|_F^2-2\left<\X^{(t)}-\R^{(t)},\R^{(t)}-\X^{\ast}\right>\notag\\
	&+\|\R^{(t)}-\X^{\ast}\|_F^2\\
	=&\|\R^{(t)}-\X^{(t)}\|_F^2-\|\R^{(t)}-\X^{\ast}\|_F^2\\
	=&\frac{\|\R^{(t)}\|_F^2-\|\Z^{(t)}\|_F^2}{\frac{\|\R\|_F^2-\|\Z^{(t)}\|_F^2}{\|\Z^{(t)}\|_F^2}\alpha_t^2+(1-\alpha_t)^2}-\|\R^{(t)}-\X^{\ast}\|_F^2\\
	\xrightarrow[]{\text{a.s.}}&\frac{n v_t-\lambda(1+\frac{1}{\rho})n v_t-\lambda n v_t^2\Delta_2}{\frac{v_t-\lambda(1+\frac{1}{\rho})v_t-\lambda v_t^2\Delta_2}{1+\lambda(1+\frac{1}{\rho})v_t+\lambda v_t^2\Delta_2}(\alpha(v_t))^2\!+\!(1\!-\!\alpha(v_t))^2}\!-\!nv_t
\end{align}\ES
where (\ref{errorxxstar}b) is from (\ref{orthoproperty2}b), (\ref{errorxxstar}c) follows from (\ref{lemma3proof}), and (\ref{errorxxstar}d) follows from (\ref{distance3}) and (\ref{distance2}) and Assumption 2. Therefore, (\ref{mseb}) holds, which concludes the proof of Theorem 2.


\bibliographystyle{IEEEtran}

\end{document}